\documentclass[11pt]{article} 
\usepackage[top=1in, bottom=1in, left=1in, right=1in]{geometry}
\usepackage{amsmath, amsfonts, amssymb, mathrsfs, setspace, graphicx, epsfig, amsthm, mhchem, caption,subcaption}

\newtheorem{theorem}{Theorem}

\newtheorem{definition}{Definition}
\newtheorem{example}{Example}

 {\begin{list}{}%
         {\setlength{\leftmargin}{#1}}%
         \item[]%
 }
 {\end{list}}

\title{Meanings and Applications of Structure in \\Networks of Dynamic Systems}
\author{Vasu Chetty and Sean Warnick\\Information and Decision Algorithms Laboratories\\Department of Computer Science\\Brigham Young University}
\date{}

\begin{document}
\maketitle

\section{Introduction}

Dynamics and structure are two of the most important properties of a system.  A system's dynamics describe its {\em behavior}, that is, how it constrains allowed combinations of manifest variables and defines what is possible and impossible \cite{willems}.  So, for example, a deterministic signal-processing system restricts the allowed output trajectory that corresponds to a given input trajectory; this allowed combination is possible, while other combinations of different output trajectories with this input trajectory are impossible.   A considerable body of literature over the last century has addressed the representation, analysis, and design of a system's dynamics, and a rich theory is now well established for doing so, especially for systems in feedback \cite{lst,feedbacksystems,paganini}.   

This focus on feedback highlights the power of systems theory to describe problems involving the interconnection of systems, and thus to describe and address questions of {\em structure}.  Indeed, results demonstrating how a feedback structure can systematically compensate for model uncertainty are among the most important system-design concepts available.  Understanding structure can also help describe information constraints in cyber-physical systems, that is, constraints characterizing what information is available to which parts of a system at various times.  Such constraints are not only useful for characterizing distributed systems, but they also play an important role in describing uncertainty about a system, including the uncertainty arising from possibilities of cyber attacks and other security problems.  Furthermore, understanding the specific structure of a particular system can play an important role in reverse engineering how the system realizes its behavior--that is, how it actually achieves the observed behavior.

Systems theory offers a powerful language for interconnecting systems into composite systems, and therefore it naturally describes system structure as the interconnection of component subsystems.  Nevertheless, its varied descriptions of systems lead to varied notions of system structure, some of which are more meaningful in certain domains than others.  For example, two representations of a controlled, causal, linear time-invariant system include its impulse-response matrix (or, equivalently, its transfer function matrix) and a state space realization.  These different representations of the same system lead to different perspectives on the structure of the system.  On the one hand, the sparsity pattern of the impulse-response matrix indicates the existence of {\em paths} from each input to each output within the system, while the sparsity pattern of the state matrices reveal details about how the system {\em captures}, {\em stores}, and {\em retrieves} information.  The critical lessons are that {\em structure is a property of a mathematical model of a system, not of the system itself}, and {\em any given system will have multiple structures available for consideration and analysis}.    

This chapter reviews four notions of system structure, three of which are contextual and classic (i.e. the complete computational structure linked to a state space model, the sparsity pattern of a transfer function, and the interconnection of subsystems) and one which is relatively new (i.e. the signal structure of a system's dynamical structure function).  Although each of these structural concepts apply to the nonlinear and stochastic setting, this work will focus on linear time invariant systems to distill the key concepts and make their relationships clear.  We then discusses three applications of the newest structural form (the  signal structure of a system's dynamical structure function): network reconstruction, vulnerability analysis, and a recent result in distributed control that guarantees the synthesis of a stabilizing controller with a specified structure or proves that no such controller exists.

\subsection{What is System Structure?}
Many physical systems have a natural notion of structure; loosely speaking, it's the way things are put together.  Certainly buildings, bridges, and other edifices are excellent examples of the idea; who doesn't immediately recognize the characteristic structure of the Eiffel Tower, the Parthenon, or the Golden Gate Bridge?  These structures maintain their distinct appearances because of the ways their respective components have been assembled; the idea of structure, in these cases, is intimately connected to the interconnection pattern of components.

Certainly this notion of structure, the interconnection pattern of components, is equally meaningful for a number of other engineered systems.  Four wheels or sweptback wings are familiar structural choices for many land or air vehicles.  Clocks, watches, locks, musical instruments, incandescent light bulbs, LEDs, televisions, radios, satellites, circuit boards, server farms, and even the internet all have meaningful notions of structure characterized by the interconnection of components or subsystems.  Each of these systems results in a characteristic physical appearance and a corresponding interconnection pattern, at various scales, that we may come to think of as its ``structure." 

But what about chemical processes, ocean currents, or even the weather?  It doesn't seem to make sense to talk about the way the ``weather" is put together.  Does the {\em fluidic} aspect of these kinds of systems remove them from meaningful interpretations of structure?  Certainly the resulting amorphous quality of these systems does make it difficult to conceptualize compartmentalized subsystems and think about their interconnection.  Nevertheless, another notion of structure, based on how manifest variables affect each other, can be very meaningful for such systems.  

We call this newer notion of structure a system's {\em signal structure}, since it describes how manifest signals affect each other.  It is defined for all systems, even those with obvious component subsystems, and sometimes it coincides with the interconnection pattern of subsystems. Nevertheless, frequently--even when a system exhibits a meaningful pattern of interconnected subsystems, and especially when no such interconnection pattern exists--the signal structure offers a unique perspective on a system's internal interactions and lends insight about the system not available otherwise.  Consider, for example, the market demand for the offering of goods available at a particular store.  Prices of the goods each day may be viewed as inputs to the system, and corresponding daily sales can be seen as outputs.  Although it is difficult to imagine compartmentalized subsystems interconnected to produce the resulting market demand, it is, on the other hand, quite natural to consider how interrelatedness of products shapes the resulting sales dynamic, leading to particular dependencies among observed sales.  Signal structure captures such dependencies and can reveal elegant structure in complex systems.  

Information diffuses and flows, much like a fluid, so cyber-physical systems benefit from both subsystem and signal structural views.  In this chapter we review standard notions of system structure, including the interconnection of subsystems, and we present a detailed treatment of signal structure and its corresponding mathematical representation, the dynamical structure function.  

\subsection{Why Does System Structure Matter?}

There are many situations where a particular structure of a system directly impacts its dynamic behavior, and thus the need, for example, of a specific shape in a ship's hull, or a design of a complex freeway interchange, is well understood.  In these cases we simply choose a structure that yields a system with the desired behavior.    

Nevertheless, what about situations where very different structural choices yield exactly the same behavior, such as is frequently the case with software, electronics, and a variety of other systems?  In these cases, is one structural choice preferred over others?  What criteria should one use to evaluate different structural options when the system dynamics are otherwise equivalent?  Consider the following:
\begin{itemize}
\item {\bf Implementation Cost}.  The fact that a given transfer function has many state realizations, some of which may be much more sparse than others, illustrates the important idea that the same manifest behavior of a system can often be realized from implementations with significantly different numbers of internal components.  In situations where the number of components is proportional to the cost of the implementation, as is the case for many physical systems, implementation cost then becomes an important reason for understanding the structural choices available to realize a specific dynamic design.     
\item {\bf Understandability}.  Internal structure of a system can be important to help one understand (or hinder an outsider from understanding) how the system works.  Hierarchy and modularity of subsystems are examples of methods for organizing designs so that complex systems can be more easily understood.  This understandability can have a major impact on other aspects of system management, such as making the system easier to: 
\begin{itemize}
\item visualize,
\item promote situational awareness,
\item verify,
\item diagnose for component failure, 
\item facilitate targeted access to system components, and 
\item maintain.  
\end{itemize}
On the other hand, making structural choices that reduce the understandability of a system can help to secure the system from various types of infiltration, including:
\begin{itemize}
\item espionage, or 
\item sabotage.
\end{itemize} 
\item {\bf Learning From Data}.  Since different mathematical models can describe the same system, and some of these models correspond to more detailed notions of structure than others, models with coarser structural descriptions are easier to learn from data than others.  Choosing a mathematical representation of a system consistent with the information available to identify it from data is critical to accurately infer its network structure; input-output data can identify the sparsity structure of a system's transfer function, but more information about the system must be known, a priori, to identify more detailed models.  Sometimes identifying structurally richer models of a system is called {\em reverse engineering}, {\em network inference} or {\em network reconstruction}.  More about network reconstruction is discussed in Section \ref{sec:netrecon}. 
    \item {\bf Attack Modeling}.   Different mathematical models of a system exhibit different notions of system structure, and these various structures expose different parts of the system as ``links" or ``nodes" in an appropriately defined structural graph.  Under the assumption that the likelihood of attacks involving multiple links in the system is inversely proportional to the number of participating links (e.g. that single-link attacks or failures are more likely than coordinated, multi-link strikes), a particular structural representation of the system also characterizes a class of anticipated perturbations with associated risks.  Analyzing the robustness of system properties such as stability, controllability, and observability with respect to perturbations in this class leads to a method for systematically characterizing system vulnerability.  More on this topic of vulnerability analysis is discussed in Section \ref{sec:vulnerability}.  
\item {\bf Constraint Modeling}.  As the size and complexity of engineered systems grow, the need to make judicious choices about how to move information from one part of the system to another becomes increasingly important, since communication costs or delays may have a significant impact on system performance.  These choices suggest the need for structural analyses of the system, and various notions of structure can effectively model different types of information constraints of the system. More on this topic, with respect to the design of distributed stabilizing controllers, is discussed in Section \ref{sec:distcontdes}.    
\end{itemize}

Although there are many purposes for developing a rich theory of system structure, this handful of reasons describe much of the motivation behind the work presented in this chapter.  The next section reviews classical notions of system structure as a context for introducing dynamical structure functions and the signal structure.  It's followed by a discussion of three motivating applications of signal structure: network reconstruction, vulnerability analysis, and the design of stabilizing, distributed controllers.    

\section{Mathematical Representations of Systems and Structures}

This section describes four different mathematical representations of systems and their structures: the state space model with its complete computational structure, the transfer function and the input-output sparsity structure, structured linear fractional transformations and the subsystem structure, and the dynamical structure function with its signal structure.  Each of these system representations completely characterize the dynamic behavior of the system.  Nevertheless, they retain varying degrees of structural information. 

In this work, a ``structure" is a directed graph.  We will see that different system representations specify different structural graphs, and each structural graph carries with it a unique interpretation, or meaning.  We will restrict our attention to finite-dimensional, causal, deterministic linear time invariant (LTI) systems defined over continuous time, but the concepts extend naturally to the nonlinear and stochastic settings with different types of independent variable.



\subsection{State Space Models and the Complete Computational Structure}

{\em State space models} are the most structurally informative system representation considered here.  The standard state space model is given by:
\begin{equation}
\begin{array}{cll} \dot{x} & = & Ax + Bu \\ y & = & Cx + Du, \end{array} \label{eq:ssequation}
\end{equation}
where $x(t)\in\mathbb{R}^n$ represents the states within the system defined over $t\in\mathbb{R}$; $\dot{x}(t)\in\mathbb{R}^n$ represents the time derivative of these state variables;  $u(t)\in\mathbb{R}^m$ are controlled inputs into the system; and $y(t)\in\mathbb{R}^p$ are measured outputs. This representation is sufficiently detailed to completely characterize both the transfer function and the dynamical structure function of a system, with their corresponding structures. 

Nevertheless, this standard state space model does not differentiate between systems with different subsystem structures.  For example, consider two systems in feedback.  One can easily compute the closed-loop dynamics of such an interconnection and represent them with a single standard state space model.  Nevertheless, if presented with this closed-loop model, one can not determine what the two subsystems were that generate it.  This failure to distinguish different subsystem structures comes from the standard state space model's lack of representation power to distinguish between the composition of functions (see Example \ref{exmpl:feedback}).

To distinguish different subsystem structures, we need to differentiate between equivalent computations such as 1) $f(x) = x$, 2) $f(x) = 2(0.5x)$ and 3) $f(x) = 0.3x + 0.7x$.  We accomplish this by introducing {\em auxiliary variables}, $w$ that represent intermediate stages of computation.  In this way we can differentiate 1) $f(x)=x$ from 2) $f(x) = 2w$ and $w=0.5x$ or 3) $f(x) = w_1 + w_2$ and $w_1 = 0.3x$ and $w_2=0.7x$, since each of these different ways of computing the same functional relationship involve zero, one, or two auxiliary variables, respectively.  The auxiliary variables that are specified, say, in a system's ``blueprint" or manifest directly to observers, help us distinguish the system's actual computational structure from others we could imagine. 

Introducing auxiliary variables into the standard state space model characterizes a differential-algebraic system of equations capable of characterizing all three of the other system representations discussed here.  We call this modified system of equations the {\em generalized state space model} of a system, and represent it as
\begin{equation} 
\begin{array}{rcl}\dot{x} & = & Ax + \hat{A}w + Bu \\ w & = & \bar{A}x+\tilde{A}w+ \bar{B}u \\ y & = & Cx + \bar{C}w+Du \end{array} \label{eq:genss} 
\end{equation}
where $w\in\mathbb{R}^{l}$, $\hat{A}\in\mathbb{R}^{n\times l}$, $\bar{A}\in\mathbb{R}^{l\times n}$, $\tilde{A}\in\mathbb{R}^{l\times l}$, $\bar{B}\in\mathbb{R}^{l\times m}$, and $\bar{C}\in\mathbb{R}^{p\times l}$.  
The number of auxiliary variables, $l$, is called the {\em intricacy} of the generalized state space model.  Choosing $\tilde{A}$ so that $I-\tilde{A}$ is invertible yields a differentiability index of zero.  This ensures that the auxiliary variables can always be algebraically eliminated from the system, producing a dynamically equivalent standard state space model (\ref{eq:ssequation}).  We call this equivalent standard state space model the {\em zero-intricacy} realization or representation of a given generalized state space model (\ref{eq:genss}).

\begin{example}
\label{exmpl:feedback}
Consider the feedback interconnection of two systems, given by
\[
\begin{array}{rclccrcl}
\dot{x}_1&=&A_1x_1+B_1r_1&\;\;\;\;\;\;\;\;\;\;\;\;\;\;\;\;&&\dot{x}_2&=&A_2x_2+B_2r_2\\
y_1&=&C_1x_1&&&y_2&=&C_2x_2
\end{array}
\] 
with $r_1 = u_1 + y_2$ and $r_2=u_2+y_1$, where $u_1$ and $u_2$ are exogenous inputs to the closed-loop system, and $y_1$ and $y_2$ are measured outputs from the closed-loop system.  Defining $w_1=y_1$ and $w_2=y_2$, we obtain the following generalized state space model of the feedback interconnection:
\begin{equation}
\begin{array}{rcrcrcr}
\left[\begin{array}{c}\dot{x}_1\\\dot{x}_2\end{array}\right]&=&\left[\begin{array}{cc}A_1&0\\0&A_2\end{array}\right]\left[\begin{array}{c}x_1\\x_2\end{array}\right]&+&\left[\begin{array}{cc}0&B_1\\B_2&0\end{array}\right]\left[\begin{array}{c}w_1\\w_2\end{array}\right]&+&\left[\begin{array}{cc}B_1&0\\0&B_2\end{array}\right]\left[\begin{array}{c}u_1\\u_2\end{array}\right]\\\\
\left[\begin{array}{c}w_1\\w_2\end{array}\right] &=& \left[\begin{array}{cc}C_1&0\\0&C_2\end{array}\right]\left[\begin{array}{c}x_1\\x_2\end{array}\right]&+&\left[\begin{array}{cc}0&0\\0&0\end{array}\right]\left[\begin{array}{c}w_1\\w_2\end{array}\right]&+&\left[\begin{array}{cc}0&0\\0&0\end{array}\right]\left[\begin{array}{c}u_1\\u_2\end{array}\right]\\\\
\left[\begin{array}{c}y_1\\y_2\end{array}\right]&=&\left[\begin{array}{cc}0&0\\0&0\end{array}\right]\left[\begin{array}{c}x_1\\x_2\end{array}\right]&+&\left[\begin{array}{cc}1&0\\0&1\end{array}\right]\left[\begin{array}{c}w_1\\w_2\end{array}\right]&+&\left[\begin{array}{cc}0&0\\0&0\end{array}\right]\left[\begin{array}{c}u_1\\u_2\end{array}\right]
\end{array}
\label{eq:generalform}
\end{equation}
Note that $I-\tilde{A}$ is invertible, thus enabling us to easily eliminate $w$ from the equations.  Doing so yields the zero-intricacy representation of the feedback interconnection:
\begin{equation}
\begin{array}{rcl}
\left[\begin{array}{c}\dot{x}_1\\\dot{x}_2\end{array}\right]&=&\left[\begin{array}{cc}A_1&B_1C_2\\B_2C_1&A_2\end{array}\right]\left[\begin{array}{c}x_1\\x_2\end{array}\right]+\left[\begin{array}{cc}B_1&0\\0&B_2\end{array}\right]\left[\begin{array}{c}u_1\\u_2\end{array}\right]\\
\left[\begin{array}{c}y_1\\y_2\end{array}\right]&=&\left[\begin{array}{cc}C_1&0\\0&C_2\end{array}\right]\left[\begin{array}{c}x_1\\x_2\end{array}\right]
\end{array}
\label{eq:zerointricacy}
\end{equation}
Although these representations are dynamically equivalent, meaning that (\ref{eq:generalform}) and (\ref{eq:zerointricacy}) generate identical state and output trajectories if they are given the same initial condition $x_o$ and input trajectory $u(t)$, (\ref{eq:generalform}) encodes information to uniquely specify the original subsystems and their feedback interconnection structure, while (\ref{eq:zerointricacy}) does not.  
\end{example}

Example \ref{exmpl:feedback} illustrates a generalized state space model and the corresponding zero-intricacy realization of a system composed of the interconnection of multiple subsystems.  In fact, whenever $I-\tilde{A}$ is invertible, every generalized state space model has a unique, well-defined zero-intricacy realization.  Likewise, every zero-intricacy state space model is dynamically equivalent to a rich variety of generalized state space models of any positive intricacy; these generalized state space models differ only in how their computations are performed, or in their underlying computational structure.  We call this structure of the most refined generalized state space description of a system, even zero-intricacy ones, the {\em complete computational structure}, and all other notions of system structure discussed in this work can be derived directly from it.    

\begin{definition}[Complete Computational Structure]
Given a generalized state space model, as in (\ref{eq:genss}), its {\em complete computational structure} is a weighted directed graph, $\mathscr{C}$ with vertex set $V(\mathscr{C})$ and edge set $E(\mathscr{C})$ given by:
\begin{itemize}
\item $V(\mathscr{C}) = \{u_1, ..., u_m, x_1, ...,x_n, w_1, ..., w_l, y_1, ..., y_p\}$, and
\item $E(\mathscr{C})$ is specified by the nonzero entries of the adjacency matrix $\mathscr{A}(\mathscr{C})$, where
\begin{equation}
\mathscr{A}(\mathscr{C}) = \left[\begin{array}{cccc} 0 & 0 & 0 & 0\\B & A & \hat{A} & 0\\ \bar{B} & \bar{A} & \tilde{A} & 0\\ D &  C & \bar{C} & 0\end{array}\right]^T.
\label{eq:ad}
\end{equation}
That is to say, a potential edge from $v_i\in V(\mathscr{C})$ to $v_j\in V(\mathscr{C})$ has weight $\mathscr{A}(\mathscr{C})_{ij}$, but we only recognize the existence of edges with non-zero weight.
\end{itemize}
\end{definition}

The generalized state space model (\ref{eq:genss}) encodes information about how the system performs the computations necessary to realize its dynamic behavior.  It is like an {\em information blueprint} of how specific components are interconnected to access information from input signals; how this information is represented (in a specific coordinate system) and combined with other data retrieved from memory; how these new calculations are stored; and how all of this data combines to produce measurable output signals. The meaning, then, of the complete computational structure characterized by (\ref{eq:ad}), is the {\em information architecture} of a very specific computation system: how information is represented, transformed, and flows through the system.  Note that there is a distinction between ``physical structure" and state space models; in some cases, the particular basis specified by a state space model is more detailed than the physical structure may suggest.  For example, consider an inertial mass.  This mass behaves like a second order system according to Newton's Second Law of Motion, but it is not clear whether states of the system are necessarily position and velocity, or whether they are some linear combinations of position and velocity.  Exactly how some systems represent and store information may be unclear, but if it were known, state space models are capable of representing this refined level of structural knowledge. These models (the generalized state space model and its associated complete computational structure) then become the most refined knowledge of our system, ground truth from which all other representations can be compared.    

Note that because intricacy variables can always be eliminated from a generalized state description without changing its dynamics, the most refined generalized state space model, with intricacy $l>0$, immediately defines a particular sequence of state space models indexed by their intricacies, $l-1,...,0$.  Each of these coarser models has a structure associated with it that we call a {\em computational structure}, but we reserve the descriptor, {\em complete computational structure} for the most refined structural specification of the system; once the complete computational structure is specified, even if it has zero intricacy, all other hypothetical refinements are considered fictitious while any agglomerative structure derived from it is a valid notion of structure for the system.   

\begin{example}
Making Example 1 concrete, consider the following two systems:
\[
\begin{array}{rclcrcl}
\begin{bmatrix}\dot{x}_1\\\dot{x}_2\end{bmatrix}&=&\begin{bmatrix}-1&2\\0&-2\end{bmatrix}\begin{bmatrix}x_1\\x_2\end{bmatrix}+\begin{bmatrix}2&-1\\-1&1\end{bmatrix}\begin{bmatrix}r_1\\r_2\end{bmatrix}&&\begin{bmatrix}\dot{x}_3\\\dot{x}_4\\\dot{x}_5\end{bmatrix}&=&\begin{bmatrix}-5&-4&2\\3&2&-1\\0&0&-3\end{bmatrix}\begin{bmatrix}x_3\\x_4\\x_5\end{bmatrix}+\begin{bmatrix}0&-1\\0&1\\1&0\end{bmatrix}\begin{bmatrix}r_3\\r_4\end{bmatrix}\\
\begin{bmatrix}y_1\\y_2\end{bmatrix}&=&\begin{bmatrix}1&2\\1&1\end{bmatrix}\begin{bmatrix}x_1\\x_2\end{bmatrix},&&\begin{bmatrix}y_3\\y_4\end{bmatrix}&=&\begin{bmatrix}1&2&0\\1&1&0\end{bmatrix}\begin{bmatrix}x_3\\x_4\\x_5\end{bmatrix},
\end{array}
\] 
interconnected in feedback, so that
\[
\begin{array}{rclcrcl}
\begin{bmatrix}r_1\\r_2\end{bmatrix}&=&\begin{bmatrix}y_3\\y_4\end{bmatrix}+\begin{bmatrix}u_1\\u_2\end{bmatrix},&\;\;\;\;\;\;\;\;\;\;&
\begin{bmatrix}r_3\\r_4\end{bmatrix}&=&\begin{bmatrix}y_1\\y_2\end{bmatrix}+\begin{bmatrix}u_3\\u_4\end{bmatrix},
\end{array}
\]
leading to the following generalized state space model:
\begin{equation}\small
\begin{array}{rcrcrcr}\begin{bmatrix}\dot{x}_1 \\ \dot{x}_2 \\ \dot{x}_3 \\ \dot{x}_4 \\ \dot{x}_5 \end{bmatrix}& = &\begin{bmatrix} -1 & 2 & 0 & 0 & 0 \\ 0 & -2 & 0 & 0 & 0 \\ 0 & 0 & -5 & -4 & 2 \\ 0 & 0 & 3 & 2 & -1\\ 0 & 0 & 0 & 0 & -3 \end{bmatrix}\begin{bmatrix}x_1 \\ x_2 \\ x_3 \\ x_4 \\ x_5 \end{bmatrix}& + &\begin{bmatrix}0 & 0 & 2 & -1\\0 & 0 & -1 & 1\\0 & -1 & 0 & 0\\0 & 1 & 0 & 0\\1 & 0 & 0 & 0 \end{bmatrix}\begin{bmatrix}w_1 \\ w_2 \\ w_3 \\ w_4\end{bmatrix}& + &\begin{bmatrix}2 & -1 & 0 & 0\\-1 & 1 & 0 & 0\\0 & 0 & 0 & -1\\0 & 0 & 0 & 1\\0 & 0 & 1 & 0 \end{bmatrix}\begin{bmatrix}u_1 \\ u_2 \\ u_3 \\ u_4\end{bmatrix} \\

\begin{bmatrix}w_1 \\ w_2 \\ w_3 \\ w_4 \end{bmatrix}& = &\begin{bmatrix} 1 & 2 & 0 & 0 & 0 \\ 1 & 1 & 0 & 0 & 0 \\ 0 & 0 & 1 & 2 & 0 \\ 0 & 0 & 1 & 1 & 0 \end{bmatrix}\begin{bmatrix}x_1 \\ x_2 \\ x_3 \\ x_4 \\ x_5 \end{bmatrix}& + &\begin{bmatrix}0 & 0 & 0 & 0\\0 & 0 & 0 & 0\\0 & 0 & 0 & 0\\0 & 0 & 0 & 0\end{bmatrix}\begin{bmatrix}w_1 \\ w_2 \\ w_3 \\ w_4\end{bmatrix}& + &\begin{bmatrix}0 & 0 & 0 & 0\\0 & 0 & 0 & 0\\0 & 0 & 0 & 0\\0 & 0 & 0 & 0\end{bmatrix}\begin{bmatrix}u_1 \\ u_2 \\ u_3 \\ u_4\end{bmatrix} \\

\begin{bmatrix}y_1 \\ y_2 \\ y_3 \\ y_4 \end{bmatrix}& = &\begin{bmatrix}0 & 0 & 0 & 0 & 0 \\ 0 & 0 & 0 & 0 & 0 \\ 0 & 0 & 0 & 0 & 0 \\ 0 & 0 & 0 & 0 & 0 \end{bmatrix}\begin{bmatrix}x_1 \\ x_2 \\ x_3 \\ x_4 \\ x_5 \end{bmatrix}& + &\begin{bmatrix}1 & 0 & 0 & 0\\0 & 1 & 0 & 0\\0 & 0 & 1 & 0\\0 & 0 & 0 & 1\end{bmatrix}\begin{bmatrix}w_1 \\ w_2 \\ w_3 \\ w_4\end{bmatrix}& + &\begin{bmatrix}0 & 0 & 0 & 0\\0 & 0 & 0 & 0\\0 & 0 & 0 & 0\\0 & 0 & 0 & 0\end{bmatrix}\begin{bmatrix}u_1 \\ u_2 \\ u_3 \\ u_4\end{bmatrix}.  \label{eq:exgssm} 
\end{array}
\end{equation}
The complete computational structure of this system, given in (\ref{eq:exgssm}), is shown in Figure \ref{fig:genccs}.
\begin{figure}[h!] \centering \includegraphics[page=1,width=.7\textwidth]{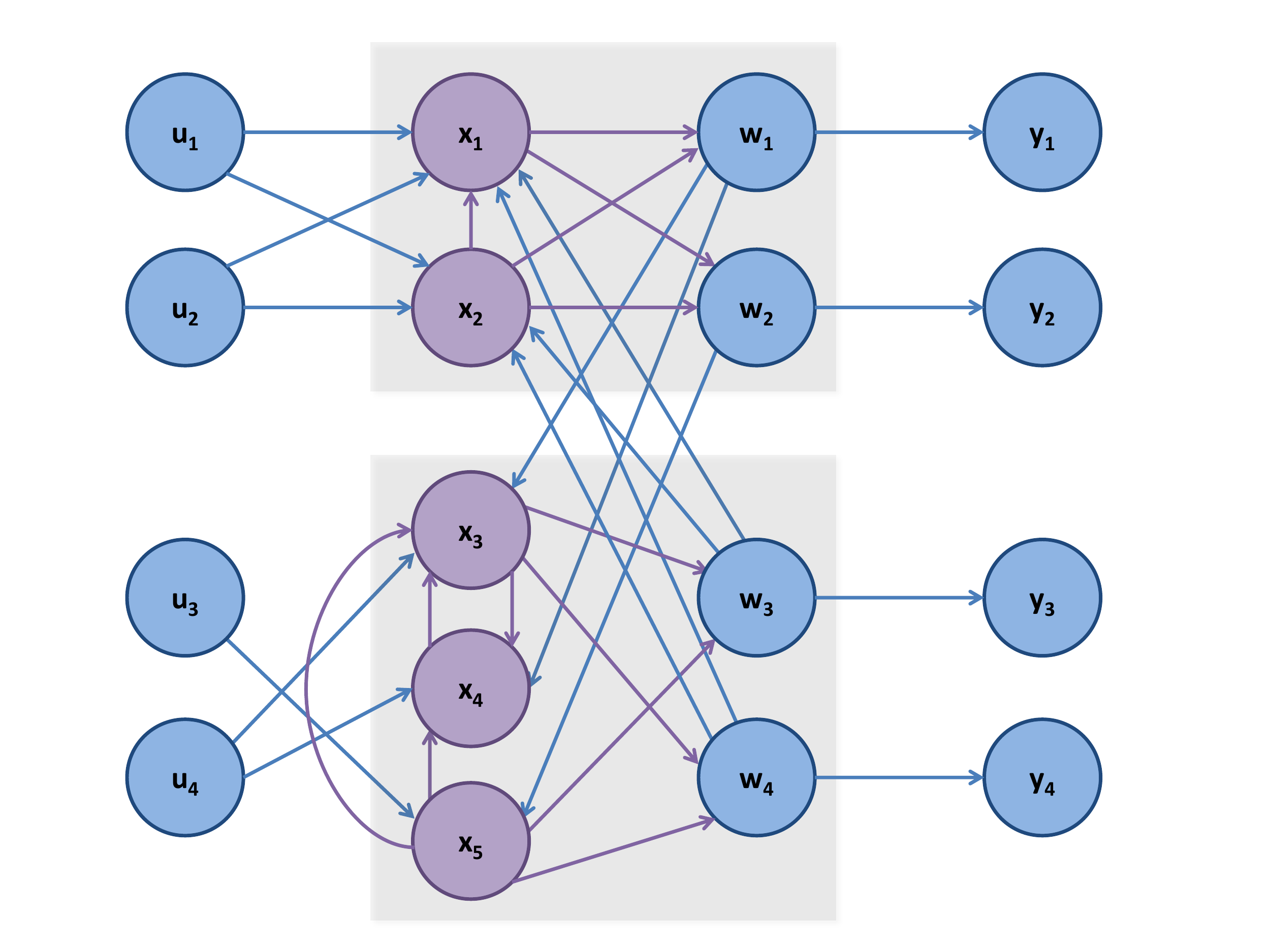} \caption{Complete computational structure of the generalized state space model from (\ref{eq:exgssm}).  Blue nodes are manifest variables, while purple nodes indicate hidden variables.  Notice that the original feedback structure of subsystems, reflected by gray boxes, is preserved, since the only interaction between subsystems is through manifest variables.} \label{fig:genccs} \end{figure}
The zero-intricacy realization of this generalized state space model, (\ref{eq:exgssm}), is then given by:
\begin{equation} 
\begin{array}{rcrcr}
\begin{bmatrix}\dot{x}_1 \\ \dot{x}_2 \\ \dot{x}_3 \\ \dot{x}_4 \\ \dot{x}_5 \end{bmatrix} & = & \begin{bmatrix}-1 & 2 & 1 & 3 & 0 \\ 0 & -2 & 0 & -1 & 0 \\ -1 & -1 & -5 & -4 & 2 \\ 1 & 1 & 3 & 2 & -1 \\ 1 & 2 & 0 & 0 & -3\end{bmatrix}\begin{bmatrix}x_1 \\ x_2 \\ x_3 \\ x_4 \\ x_5\end{bmatrix} & + & \begin{bmatrix}2 & -1 & 0 & 0\\-1 & 1 & 0 & 0\\0 & 0 & 0 & -1\\0 & 0 & 0 & 1\\0 & 0 & 1 & 0 \end{bmatrix}\begin{bmatrix}u_1 \\ u_2 \\ u_3 \\ u_4 \end{bmatrix} \\ 
\begin{bmatrix}y_1 \\ y_2 \\ y_3 \\ y_4 \end{bmatrix} & = & \begin{bmatrix}1 & 2 & 0 & 0 & 0 \\ 1 & 1 & 0 & 0 & 0  \\ 0 & 0 & 1 & 2 & 0 \\ 0 & 0 & 1 & 1 & 0 \end{bmatrix}\begin{bmatrix}x_1 \\ x_2 \\ x_3 \\ x_4 \\ x_5 \\ x_6\end{bmatrix} & + & \begin{bmatrix}0 & 0 & 0 & 0\\0 & 0 & 0 & 0\\0 & 0 & 0 & 0\\0 & 0 & 0 & 0\end{bmatrix}\begin{bmatrix}u_1 \\ u_2 \\ u_3 \\ u_4\end{bmatrix} \end{array} \label{eq:exzissm}
\end{equation}
The computational structure of the zero intricacy realization, given in (\ref{eq:exzissm}), is shown in Figure \ref{fig:ziccs}.  Notice the differences with the complete computational structure shown in Figure \ref{fig:genccs}.  For example, the complete computational structure has nodes for auxiliary variables, $w$, while the computational structure of the zero intricacy realization does not.  Also, original subsystem structure is preserved in the complete computational structure, highlighted by the background gray boxes, while it is lost in the computational structure of the zero intricacy realization, resulting in no distinguishable subsystems. 
\begin{figure}[h!] \centering \includegraphics[page=8,width=.7\textwidth]{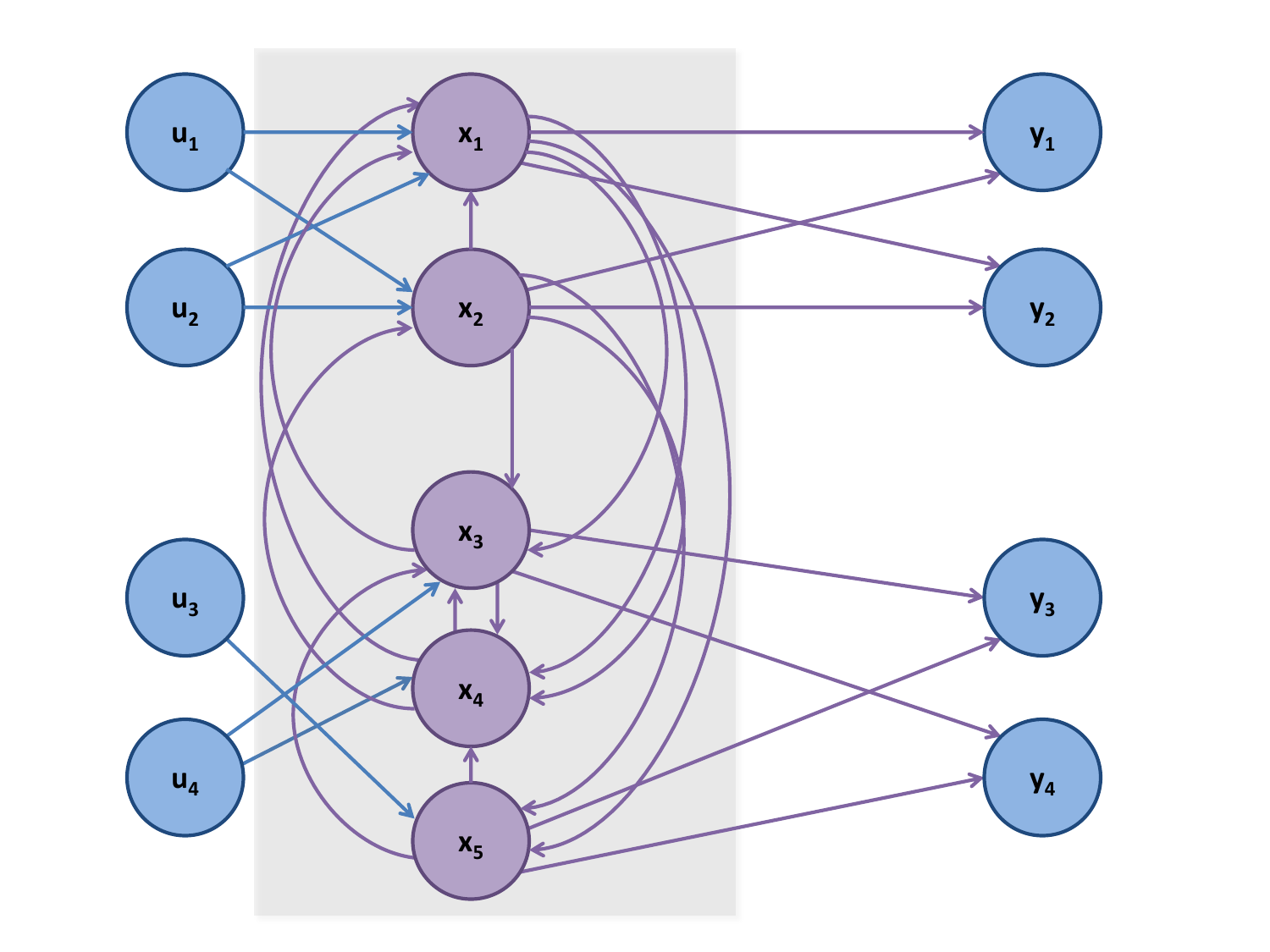} \caption{Computational structure of the zero intricacy realization (\ref{eq:exzissm}) of the generalized state space model in (\ref{eq:exgssm}).  Like Figure \ref{fig:genccs}, blue nodes indicate manifest variables while purple nodes are hidden variables.  Notice that the original subsystem structure is lost, and only a single subsystem remains visible from manifest variables.} \label{fig:ziccs} \end{figure}
\end{example}

\subsection{Functional System Descriptions and the Manifest Structure}
\label{sec:tf}

While state space models are the most structurally informative system representations, functional system descriptions, such as convolution models or transfer functions, are at the other end of the spectrum.  These ``black box" representations of a system are capable of describing the same dynamic behavior as their state space counterparts, yet they do not model the detailed interactions among system components the way state space representations do\footnote{Although transfer functions and convolution models of LTI systems assume zero initial conditions, the impact of a non-zero initial condition is easily modeled with the addition of an appropriately designed external disturbance.}.

This inability to convey detailed structural information is not necessarily a weakness, however.  For example, functional representations need fewer parameters to characterize a given dynamic behavior, making them easier to learn from data (called {\em system identification} \cite{sysid1, sysid2}) than their state space counterparts .  Moreover, their parsimonious description of a system's dynamics creates an important distinction between a system's behavior and how it realizes that behavior, enabling a concerted focus on the design of a system's dynamics without worrying about implementation.  

Just as high-level programming languages abstract many of the details of the computer they run on, functional system descriptions are high-level abstractions of state space models.  In particular, the specific processes a state realization uses to decide which information is stored in which parts of the state vector correspond to memory management activities that are completely invisible to a functional description of a system.  This distinction is further exemplified by noting that state space models are imperative descriptions of a system, encoding computations in terms of the time evolution of the system state, while functional descriptions are inherently declarative, specifying what the system does without prescribing how it should do it.   

The result of this high-level/low-level relationship between functional system descriptions and state space models is a one-to-many relationship between the two model classes.  That is, every state space model has a zero-intricacy realization as in Equation (\ref{eq:ssequation}) that identifies a unique functional system description, whether it be the impulse response matrix of a convolution model or a transfer function matrix, given by:
\begin{equation}
\begin{array}{lcl}
y(t) = h(t)*u(t)&\;\;\;\;\;&Y(s)=H(s)U(s)\\
h(t) =Ce^{At}B+D\delta(t),&\;\;\;\;\;&H(s) = C(SI-A)^{-1}B+D,
\end{array}
\label{eq:fsd}
\end{equation}
where $*$ denotes convolution, $\delta(t)$ is the Dirac delta function, $h(t)$ is the system's $p \times m$ impulse response matrix, $Y(s)$ and $U(s)$ are the Laplace transforms of $y(t)$ and $u(t)$, and $H(s)$ is the system's $p \times m$ transfer function matrix--which is also the Laplace transform of $h(t)$.  

Note, however, that there are many state space models that specify the same impulse response or transfer function; each of these state space models specifies a different implementation (or realization) of the same dynamic behavior.  Among all these state realizations of a given functional description of a system, some have fewer states than others.  In fact, systems with functional descriptions that can be described by finite-dimensional LTI state space models\footnote{Although all LTI state space models have transfer functions, not all transfer functions have state space realizations.  This is because the imperative nature of state space models demand that they are {\em causal}, meaning that future values of manifest variables only depend on past and present values of manifest variables.  Transfer functions that are {\em proper} rational functions of the Laplace variable correspond to causal finite dimensional LTI systems; we do not concern ourselves with other kinds in this work.} have a unique integer, $n$, associated with them called the {\em Smith-McMillan degree}.  This degree is the minimal number of states necessary for any state space realization of the system.  Nevertheless, even restricting attention to state space models with order equal to the Smith-McMillan degree does not yield a unique state realization; given a minimal realization $(A,B,C,D)$ of a transfer function $H(s)$, any $n \times n$ transformation, $T$, yields another minimal realization $(\hat{A},\hat{B},\hat{C},\hat{D})$ given by:
\begin{equation}
\begin{array}{cccc}
\hat{A} = TAT^{-1}&\hat{B}=TB&\hat{C}=CT^{-1}&\hat{D}=D
\end{array}
\label{eq:trans}
\end{equation}
such that $C(sI-A)^{-1}B+D = H(s) = \hat{C}(sI-\hat{A})^{-1}\hat{B}+\hat{D}$.  Thus, even among minimal realizations, there are infinitely many implementations of a given dynamic behavior specified by a functional description such as $H(s)$, and these implementations differ only in their structural properties.

The functional description of a system, however, retains only the structural properties that are common among all of its state realizations, which is precisely the mathematical structure of the functional description itself.  This structure describes the internal closed-loop relationships among manifest variables, and therefore is called the {\em manifest structure}.

\begin{definition}[Manifest Structure]
Given a generalized state space model, as in (\ref{eq:genss}), identified by a functional system description, as in (\ref{eq:fsd}), its {\em manifest structure} is a weighted directed graph $\mathscr{M}$ with vertex set $V(\mathscr M)$ and edge set $E(\mathscr M)$ given by:
\begin{itemize}
\item $V(\mathscr{M}) = \{u_1, ..., u_m, y_1, ..., y_p\}$, each representing a manifest signal of the system, and
\item $E(\mathscr{M})$ has an edge from $u_i$ to $y_j$, labeled by either $H_{ji}$ or $h_{ji}$, provided they are non-zero. 
\end{itemize}
\end{definition}
\noindent Note that when a system's manifest variables partition naturally into inputs and outputs, then its manifest structure is a bipartite graph, with directed edges from inputs to outputs. 

An alternative definition of the manifest structure characterizes $\mathscr M$ directly from $\mathscr C$ using only graphical properties (which is useful when extending these results to the nonlinear setting).  In that case, we say $\mathscr M$ has an edge from $u_i$ to $y_j$ if the net impact of all paths in $\mathscr C$ from $u_i$ to $y_j$ is non-zero, or, equivalently, if every equivalent realization of the system, specified by a transformation $T$ as in (\ref{eq:trans}), with complete computational structure $\mathscr{C}_T$, has a path from $u_i$ to $y_j$.

\begin{example}
Consider the zero-intricacy state space model in (\ref{eq:exzissm}) from Example 2.  The corresponding transfer function is given by:

\[H(s) = C(sI-A)^{-1}B+D =\] \begin{equation}\begin{bmatrix} 0 & \frac{s^2+5s+6}{s^3+6s^2+11s+5} & \frac{1}{s^3+6s^2+11s+5} & 0 \\ \frac{s+1}{s^2+3s+1} & \frac{s^3+6s^2+11s+8}{s^5+9s^4+30s^3+44s^2+26s+5} & \frac{3s+3}{s^5+9s^4+30s^3+44s^2+26s+5} & \frac{1}{s^2+3s+1} \\ \frac{1}{s^2+3s+1} & \frac{s^2+7s+10}{s^5+9s^4+30s^3+44s^2+26s+5} & \frac{2s^2+6s+5}{s^5+9s^4+30s^3+44s^2+26s+5} & \frac{s+2}{s^2+3s+1} \\ 0 & \frac{1}{s^3+6s^2+11s+5} & \frac{s+1}{s^3+6s^2+11s+5} & 0 \end{bmatrix} \label{eq:tfexample}\end{equation}
The manifest structure corresponding to this transfer function, (\ref{eq:tfexample}), that represents the internal closed-loop pathways from inputs to outputs of the system in (\ref{fig:ziccs}) is given in Figure \ref{fig:tfmanifeststructure}.

\begin{figure}[h!]
        \centering
        \begin{subfigure}[b]{0.45\textwidth}
                \includegraphics[page=2,width=\textwidth]{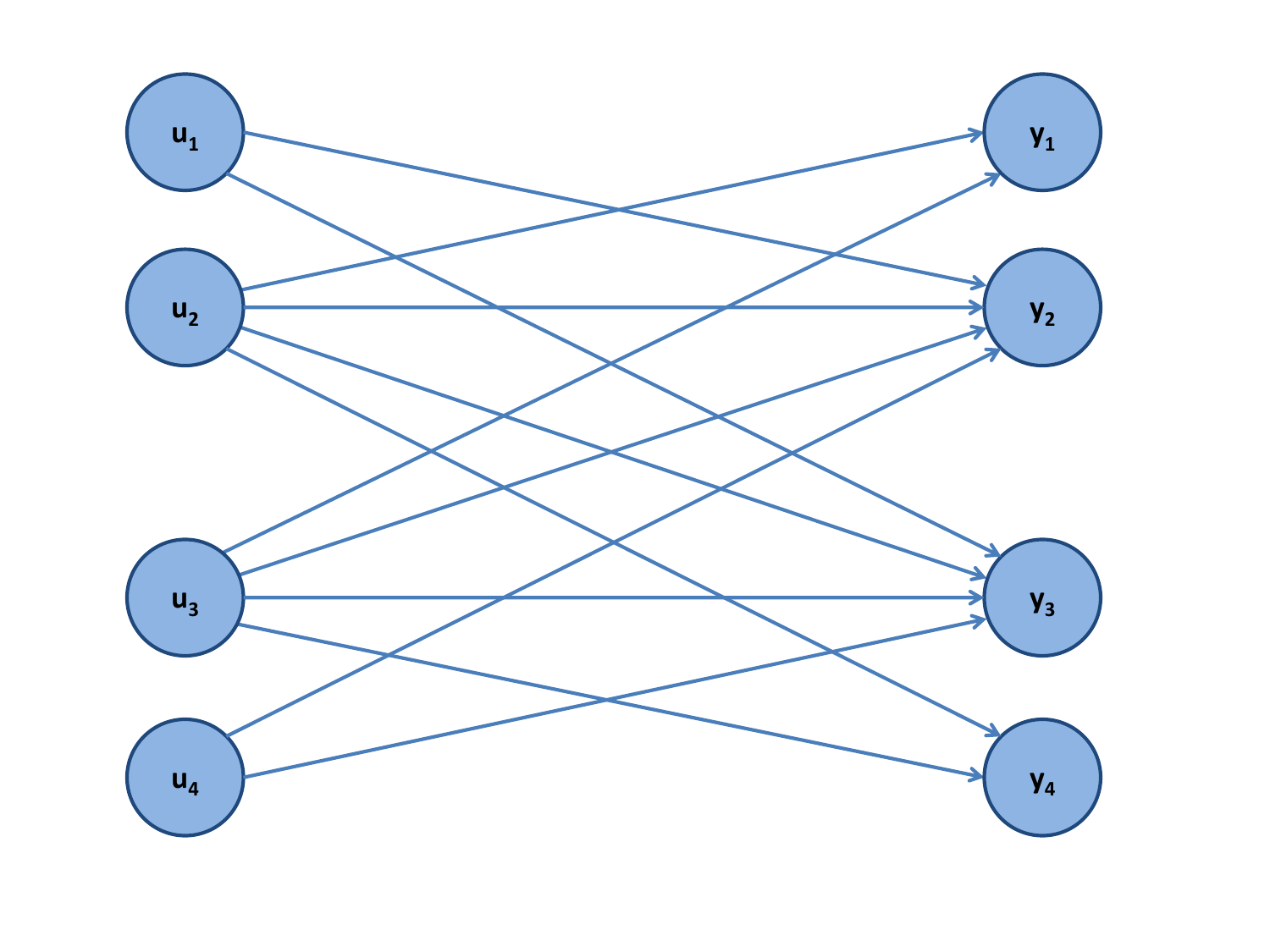}
                \caption{Manifest structure of the system with transfer function (\ref{eq:tfexample})}
                \label{fig:tfs}
        \end{subfigure}
        \begin{subfigure}[b]{0.45\textwidth}
                \includegraphics[page=3,width=\textwidth]{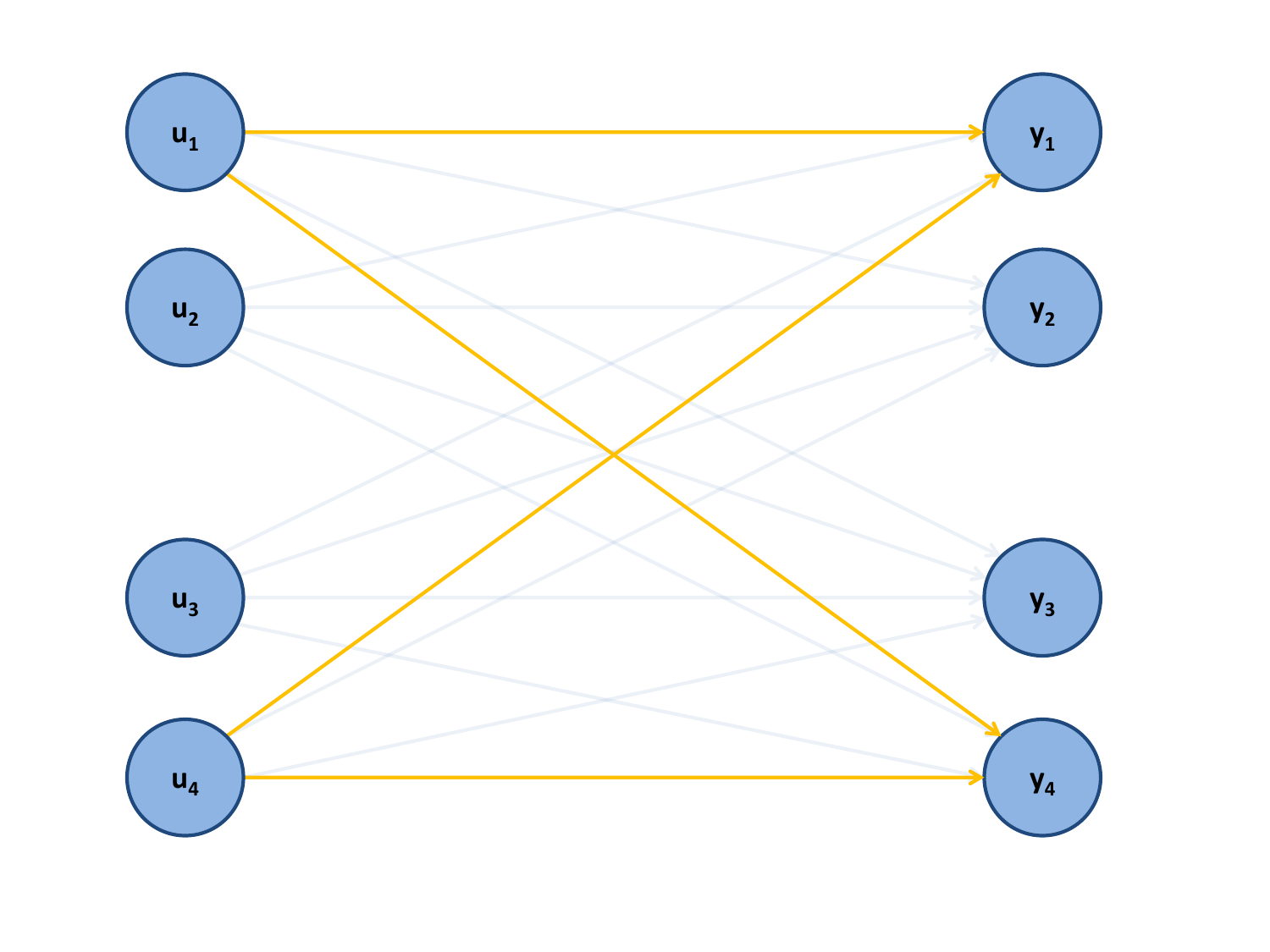}
                \caption{Missing edges in the manifest structure, corresponding to zero elements in $H$.}
                \label{fig:tfsmisspath}
        \end{subfigure}
        \caption{Manifest structure of the same system from Figures \ref{fig:genccs} and \ref{fig:ziccs}. Notice the lack of edges from $u_1$ to $y_1$ and $y_4$, and from $u_4$ to $y_1$ and $y_4$, corresponding to associated zeros in $H(s)$.  These missing links are highlighted in Figure \ref{fig:tfsmisspath}.  Note that these links are missing in the manifest structure even though paths exist in Figure \ref{fig:ziccs} from every input to every output.}\label{fig:tfmanifeststructure}
\end{figure}

Note that in some cases, although a pathway exists from an input to an output in the system's complete computational structure, it is possible that the corresponding transfer function from the input to the output is zero.  For example, notice that although paths exist from every input to every output in the computational structure of the zero intricacy realization generating $H$ (Figure \ref{fig:ziccs}), $H_{11}$, $H_{41}$, $H_{14}$, and $H_{44}$ are nevertheless all zero. Thus, we see that the existence of paths from $u_i$ to $y_j$ is not sufficient for $H_{ij}$ to be nonzero; the closed-loop, net effect of {\em all} paths from $u_i$ to $y_j$ must be nonzero for $H_{ij}$ to be nonzero; exact cancellations, which can be common in software and other engineered systems, can generate zeros in the functional description.


\end{example}

\subsection{Structured Linear Fractional Transformations and the Subsystem Structure}
\label{sec:substr}

Having identified the complete computational structure as the most informative structural representation, and the manifest structure as the least, we now explore the most common intermediate structural representation: the interconnection of subsystems.  Subsystem structure is less informative than the complete computational structure because it does not reveal the internal structure of subsystems.  On the other hand, subsystem structure can be more informative than manifest structure because it reveals the interconnection pattern among subsystems.

To isolate and represent the interconnection pattern of subsystems for a given system, begin by considering a set of $q$ subsystems, $S=\{\begin{array}{cccc}S_1&S_2&...&S_q\end{array}\}$, interconnected into a composite system, $H$.  It is conceivable that each of these subsystems are themselves divisible into constituent subsystems, or that not all of the $q$ subsystems are discernible from $H$'s manifest variables, so we specify the level of modeling abstraction by: 
\begin{enumerate}
\item Modeling each of the $q$ constituent subsystems with a suitable functional description, such as a proper or strictly proper transfer function $S_i(s)$, $i=1,2,...,q$, or a {\em single-subsystem} state space realization, characterized as a generalized state space model with subsystem structure consisting of a single subsystem, so that no further division of the subsystems is possible, and
\item Ensuring that each of the subsystem's outputs, $w_i$, is a measured output of the composite system $H$, so $y=[\begin{array}{cccc}w_1^T&w_2^T&...&w_q^T\end{array}]^T$, where $y$ is the output of $H$.
\end{enumerate}
Note that each subsystem is {\em distinct}, meaning that state variables internal to one subsystem are different from those of the other subsystems, yielding no mechanism for interaction except through their respective manifest variables. Let $u$ be a vector of external inputs; $v_i$ and $w_i$ be the vectors of inputs and outputs for system $S_i$; and $v$ and $w$ be the stacked inputs and outputs from all systems, $v=[\begin{array}{cccc}v_1^T&v_2^T&...&v_q^T\end{array}]^T$ and $w=[\begin{array}{cccc}w_1^T&w_2^T&...&w_q^T\end{array}]^T$, so that $w = y$.  Interconnecting these systems then means defining binary matrices $L$ and $K$ such that:
\begin{equation}
\left[\begin{array}{cc}L&K\end{array}\right]\left[\begin{array}{c}u\\w\end{array}\right]=v.
\end{equation}

Our convention is that the process of interconnection only allows the selection of particular signals and possibly adding them together, thus restricting the interconnection matrices, $L$ and $K$, to have elements with values of either zero or one; all other computations are part of the systems in $S$.  Further, we assume that the resulting interconnection is {\em well-posed}, meaning that all signals within $H$ are uniquely specified for any value of external inputs and underlying state variables \cite{doyle}.  This assumption ensures that the proposed interconnection is physically sensible and not merely a mathematical artifact. 


The composite system, $H$, is then clearly defined by the structured linear fractional transformation (LFT) as in Figure \ref{fig:lft}, given by:
\begin{equation}
\begin{array}{rcl}
N\left[\begin{array}{c}u\\w\end{array}\right]&=&\left[\begin{array}{c}y\\v\end{array}\right],\\
w &=& Sv,
\end{array}
\label{eq:lft}
\end{equation}
where
\begin{equation}
\begin{array}{lcr}
N = \left[\begin{array}{cc}0&I\\L&K\end{array}\right],&\;\;\;\;\;\;\;\;\;&
S = \left[\begin{array}{cccc}S_1&0&...&0\\0&S_2&&0\\\vdots&&\ddots&\vdots\\0&&...&S_q\end{array}\right]
\end{array}
\label{eq:S}
\end{equation}
and $S_i$ can be represented by either a suitable functional description, such as a proper or strictly proper transfer function matrix or the associated impulse response matrix of a convolution model, or by any single-subsystem generalized state realization.  The symbol $S$ is overloaded, representing both the set of subsystems and the decoupled operator of subsystem models in (\ref{eq:S}), but the appropriate meaning should always be clear from context. Equations (\ref{eq:lft}) and (\ref{eq:S}) characterize $H$ as a structured LFT in terms of $S$.  Combining these equations yields, for example, $Y(s) = [S(s)(I-KS(s))^{-1}L]U(s)$, implying that $H(s) = S(s)(I-KS(s))^{-1}L$, where $Y(s)$ and $U(s)$ are the Laplace transforms of $y(t)$ and $u(t)$, respectively; similarly, a well-specified expression can be obtained for $h$ directly in the time domain.  The functional description of the composite system, $H$, in either the time or frequency domain, is completely specified by the structured LFT description in (\ref{eq:lft}) and (\ref{eq:S}).

\begin{figure}[h!] \centering \includegraphics[width=0.6\textwidth]{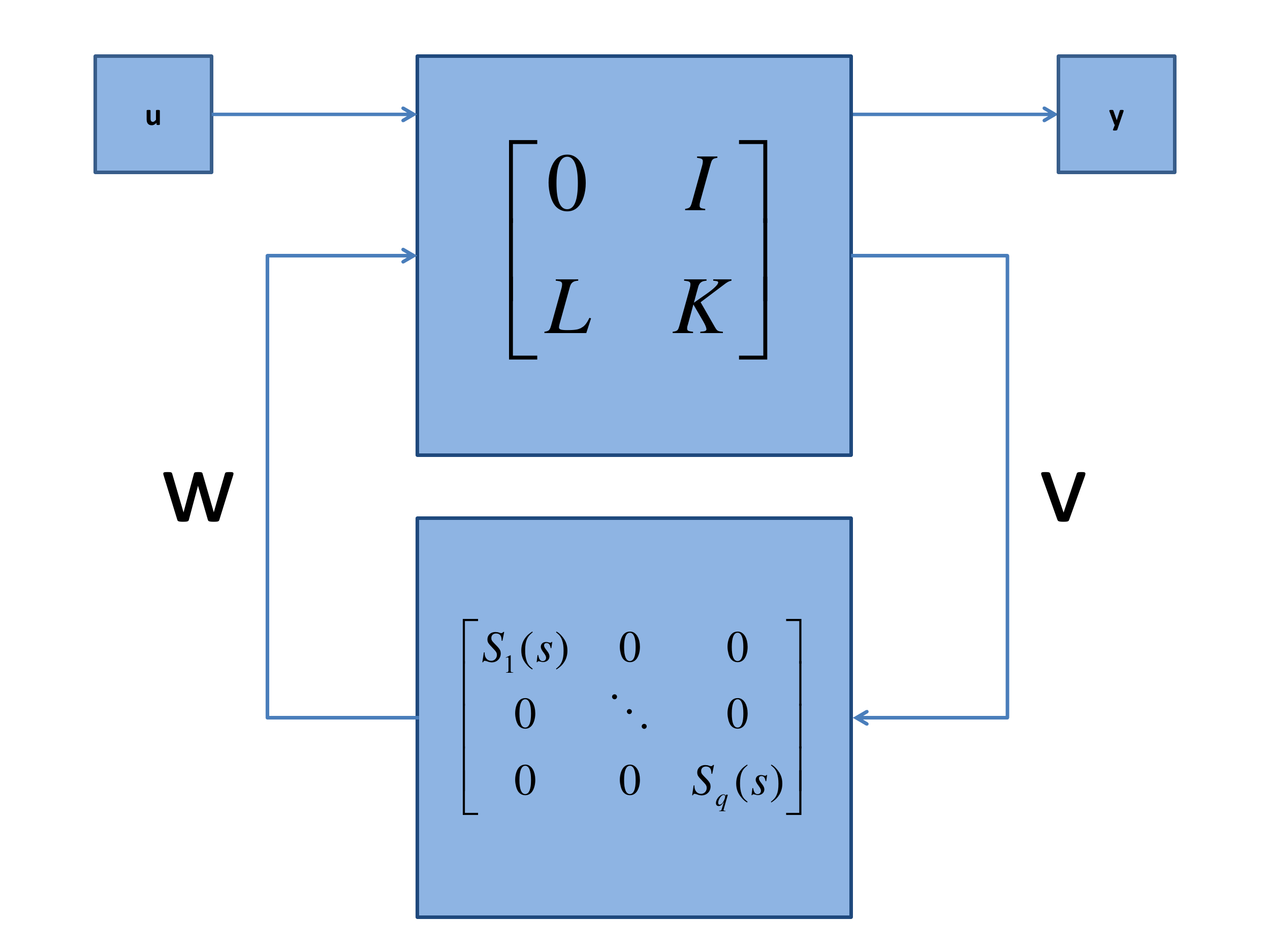} \caption{A structured linear fractional transformation revealing the interconnection structure among subsystems in binary matrices $L$ and $K$.} \label{fig:lft} \end{figure} 

Although the structured LFT completely specifies the functional description of the composite system, $H$, the structured LFT does not have enough structural information to specify $H$'s complete computational structure or its associated generalized state space description.  To do so, it would need information about the ``true" structure of each constituent subsystem.  This point may be clear when $S$ is specified by a functional description for each subsystem, such as its transfer function, but it becomes more subtle when $S$ is specified by a generalized state space model for each subsystem.  In this case, it is important to understand that the state space model for each subsystem in $S$ can be {\em any} single-subsystem realization of the associated transfer function, $S_i(s)$, since the structured LFT does not use any information about the internal structure of its subsystems.  To realize the ``true" generalized state description of $H$ an its associated complete computational structure, one must have accurate descriptions of the complete computational structures for each constituent subsystem to complement the ``interconnection" information in the structured LFT.  

The structured LFT reveals the interconnection structure among subsystems, encoded in the binary {\em interconnection matrix}, $N$, in general, and in $L$ and $K$ in particular.  Note that the interconnection structure in $N$ is unaffected by whether the subsystems in $S$ are represented by state space models or transfer functions.  The internal computational structure of subsystems, revealed by state models of subsystems but not transfer function representations of subsystems, is not used when representing the subsystem structure of a system--only the interconnection structure {\em among} subsystems, not {\em within} subsystems, is relevant for this representation.  

Aggregating $L$ and $K$ appropriately to account for the potentially multi-input multi-output nature of the constituent subsystems yields adjacency matrices from which the composite system's subsystem structure can be built.  To accomplish this, let $e_{v_i}$ denote the vector of ones with length equal to the length of vector $v_i$.  We then define the aggregation matrices
\begin{equation}
\begin{array}{ccc}
A_v = \left[\begin{array}{cccc}e_{v_1}^T&0&...&0\\0&e_{v_2}^T&&0\\\vdots&&\ddots&\vdots\\0&0&...&e_{v_q}^T\end{array}\right],&\;\;\;\;&A_w = \left[\begin{array}{cccc}e_{w_1}^T&0&...&0\\0&e_{w_2}^T&&0\\\vdots&&\ddots&\vdots\\0&0&...&e_{w_q}^T\end{array}\right],
\end{array}
\label{eq:agmat}
\end{equation}
and use them to create the adjacency matrices:
\begin{equation}
\begin{array}{ccc}
\mathscr{A}(L) = sgn(A_vL)^T,&\;\;\;\;&\mathscr{A}(K) = sgn(A_vKA_w^T)^T,
\end{array}
\label{eq:adjacency}
\end{equation}
where $sgn(\cdot)$ denotes the sign function, yielding a value of one for positive entries, zero for zero, and negative one for negative entries (which can never occur in this case).  With these definitions, we are now prepared to characterize a system's {\em subsystem structure}:
\begin{definition}[Subsystem Structure]
Given a generalized state space model, as in (\ref{eq:genss}), identified by a structured LFT, $(N, S)$, as in (\ref{eq:lft}) and (\ref{eq:S}) and with associated aggregation matrices as in (\ref{eq:agmat}) and adjacency matrices as in (\ref{eq:adjacency}) , its {\em subsystem structure} is a weighted directed graph $\mathscr{S}$ with vertex set $V(\mathscr S)$ and edge set $E(\mathscr S)$ given by:
\begin{itemize}
\item $V(\mathscr{S}) = \{u_1, ..., u_m, S_1, ..., S_q, y_1, ..., y_p\}$, representing input signals, subsystems, and output signals, respectively.
\item $E(\mathscr{S})$ has an edge from
\begin{itemize}
\item $u_i$ to $S_j$ if $\mathscr{A}(L)_{ij}=1$, labeled $u_i$;
\item $S_i$ to $S_j$ if $\mathscr{A}(K)_{ij}=1$, labeled $w_i$;
\item $S_i$ to $y_j$ if $(A_w)_{ij}=1$, labeled $y_j$.
\end{itemize} 
\end{itemize}
\label{def:sub}
\end{definition}

Note that the subsystem structure is qualitatively different from either the complete computational structure or the manifest structure in a few ways.  First, while all the nodes of either the complete computational structure or the manifest structure represent {\em signals}, the nodes of the subsystem structure represent {\em systems}, namely the subsystems and exosystems associated with the generation of each input or measurement of each output signal.  As a result, we often denote the nodes in the subsystem structure with a different shape, e.g. rectangles instead of circles, to highlight this distinction (see Figure \ref{fig:substr}).  Also, the edges in both the complete computational structure and the manifest structure are labeled to represent {\em systems}, while the edges in the subsystem structure are labeled with the names of {\em signals}.  These distinctions make it clear that the subsystem structure carries the interpretation of a {\em block diagram}, while the other structures are {\em signal flow graphs}.  


The definition of subsystem structure given above characterizes the graph in terms $N$ and $S$.  Nevertheless, the subsystem structure can be obtained directly from the complete computational structure, which not only lends a graphical interpretation to the concept of a subsystem, but naturally facilitates the extension of the definitions to the nonlinear and stochastic setting.  We achieve this by first extending the definition of a {\em manifest node} or {\em manifest signal} of $\mathscr{C}$ to include any node representing a signal identically equal to a manifest signal, $u_i$ or $y_j$.  We then consider the subgraph of $\mathscr{C}$ obtained by 1) removing all input nodes and any outgoing edges leaving them, 2) removing all output nodes and any incoming edges entering them, and 3) removing all outgoing edges leaving any remaining manifest nodes.  This subgraph, $\mathscr{H}$ is the {\em hidden structure} of $\mathscr{C}$, and it immediately reveals its subsystems and their interconnection, as follows:  

\begin{theorem}
Consider a system $H$ characterized by a structured LFT, $(N, S)$.  Construct a complete computational structure for $H$, as in (\ref{eq:genss}), by realizing each subsystem in $S$ with a single-subsystem state space model, and let $\mathscr{C}$ be the resulting complete computational structure.  Then every connected component of $\mathscr{H}$, the hidden structure of $\mathscr{C}$, corresponds to a distinct subsystem in $S$.
\label{LFTlemma}
\end{theorem}

\begin{proof}
Since each subsystem is realized by a single-subsystem state space realization, variables internal to each subsystem correspond to nodes of $\mathscr{C}$ that are connected to each other.  Moreover, since all outputs of $S$ are manifest, and $S$ is diagonal, these connected components can only be interconnected by manifest signals.  By removing all outgoing edges from internal manifest nodes in $\mathscr{C}$, as well as removing all input and output nodes and their associated edges, $\mathscr{H}$ isolates each subsystem so the remaining connected components of $\mathscr{H}$ correspond to the subsystems in $S$.    
\end{proof}
\noindent The next example illustrates this procedure of obtaining a system's subsystem structure directly from its complete computational structure.

\begin{example}
Consider the generalized state space model of two subsystems in feedback from Example 2, given by Equation (\ref{eq:exgssm}).  Figures \ref{fig:genccs} and \ref{fig:ccs} illustrate the system's complete computational structure, $\mathscr{C}$, and we can generate its subsystem structure by identifying the connected components in the hidden structure of $\mathscr{C}$, as demonstrated in Figure \ref{fig:developsubstructure}:

\begin{figure}
        \centering
        \begin{subfigure}[b]{0.45\textwidth}
                \includegraphics[page=1,width=\textwidth]{examples}
                \caption{Step 1: Identify the manifest variables in the Complete Computational Structure (shaded blue).}
                \label{fig:ccs}
        \end{subfigure}
        \begin{subfigure}[b]{0.45\textwidth}
                \includegraphics[page=4,width=\textwidth]{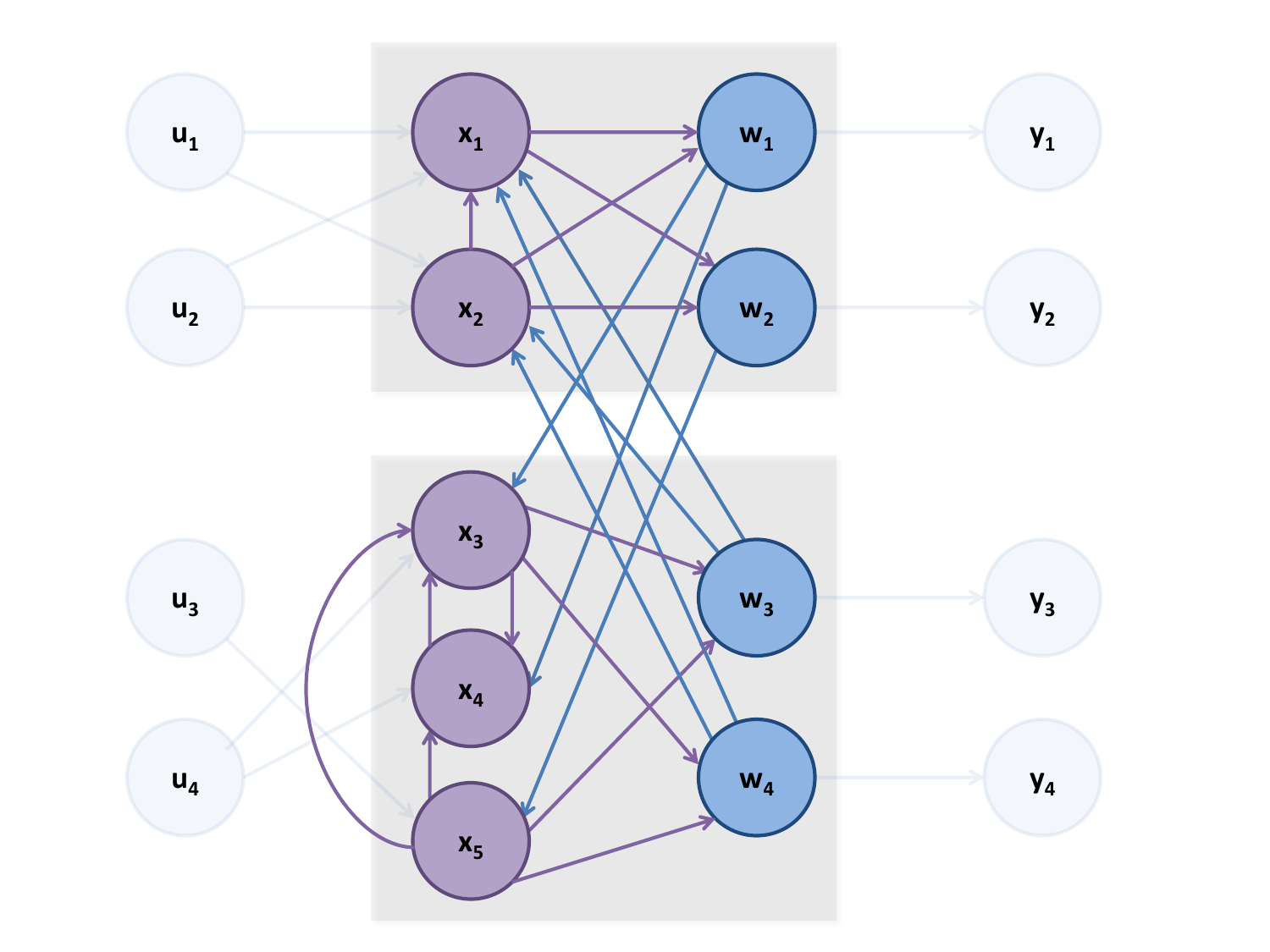}
                \caption{Step 2: Remove all input and output nodes, along with any edges adjacent to these nodes..}
                \label{fig:remnodes}
        \end{subfigure}
        \begin{subfigure}[b]{0.45\textwidth}
                \includegraphics[page=5,width=\textwidth]{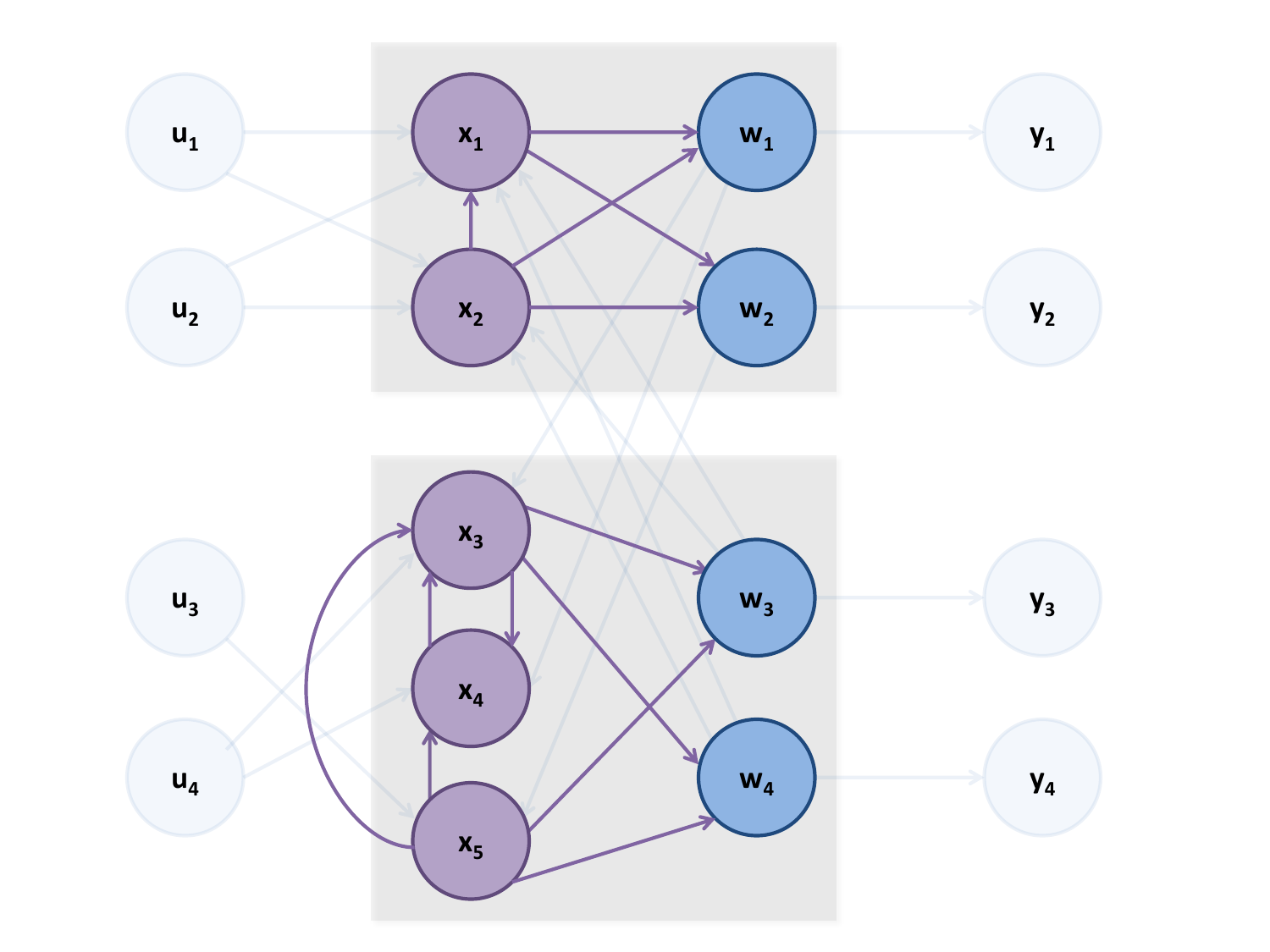}
                \caption{Step 3: Remove any outgoing edges from any remaining manifest variables.}
                \label{fig:remedges}
        \end{subfigure}
        \begin{subfigure}[b]{0.45\textwidth}
                \includegraphics[page=6,width=\textwidth]{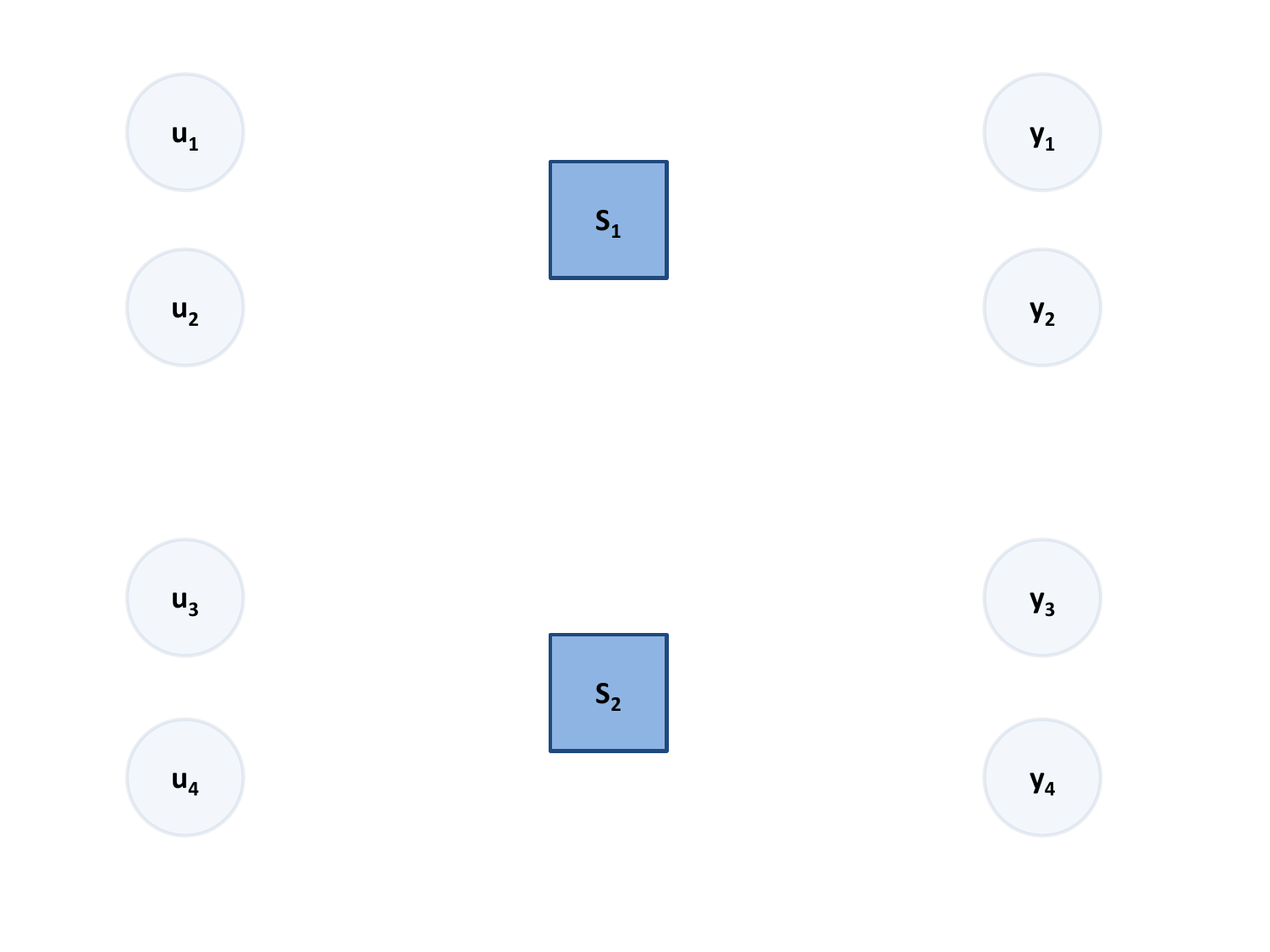}
                \caption{Step 4: Remaining connected components correspond to subsystems.}
                \label{fig:conncomp}
        \end{subfigure}
        \begin{subfigure}[b]{0.5\textwidth}
                \includegraphics[page=7,width=\textwidth]{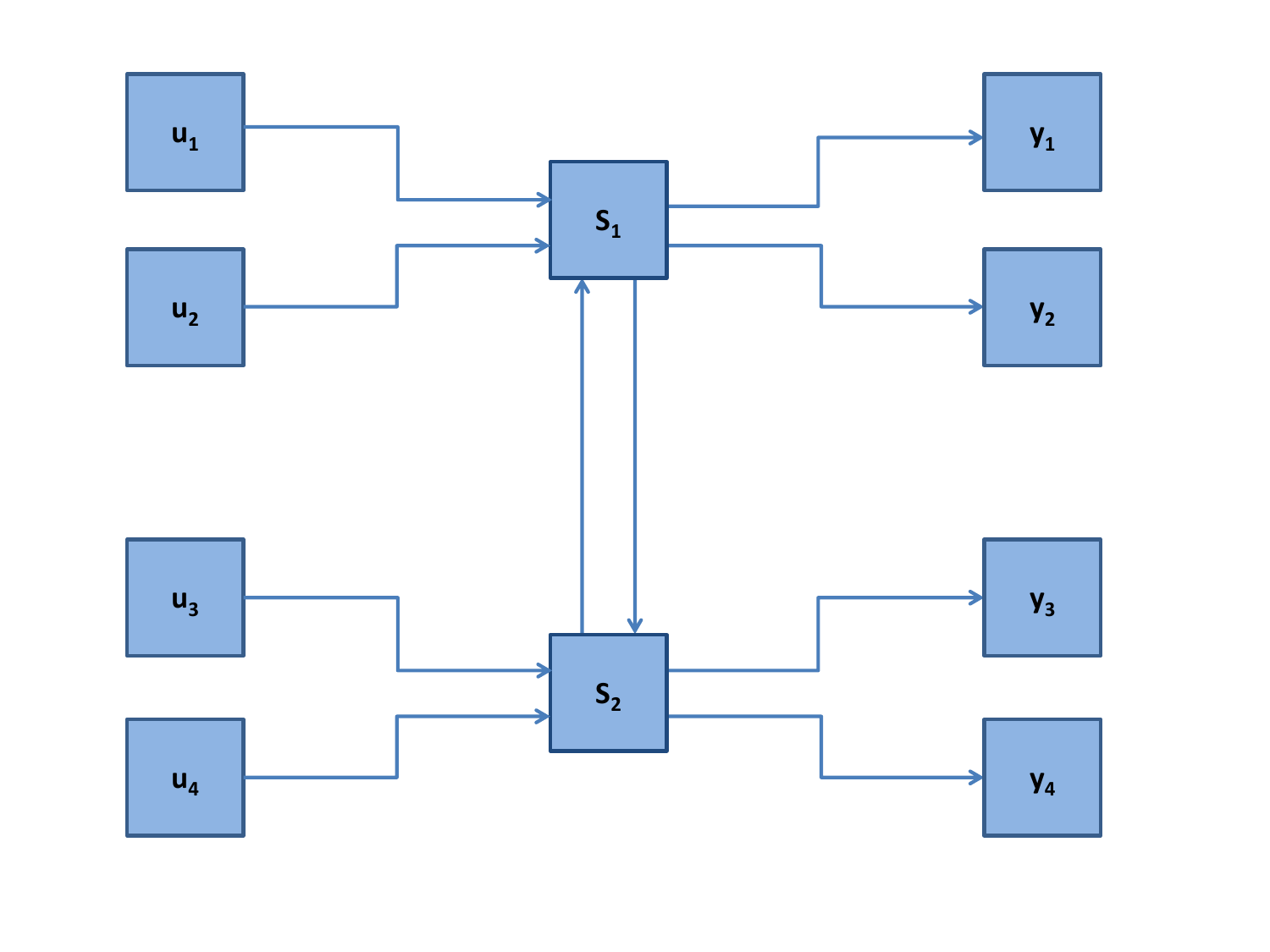}
                \caption{Step 5: Reintroduce the input and output variables as exosystem nodes, and replace all removed edges, compressing any duplicate edges into a single edge.}
                \label{fig:substr}
        \end{subfigure}
        \caption{Subsystem structure, built from a system's complete computational structure.}\label{fig:developsubstructure}
\end{figure}

The process of constructing a system's subsystem structure from its complete computational structure involves 1) identifying all manifest nodes in $\mathscr{C}$, 2) removing all input and output nodes and their adjacent edges, 3) removing all outgoing edges from any remaining manifest nodes.  These three steps construct the hidden structure, $\mathscr{H}$, and each connected component in $\mathscr{H}$ corresponds to a subsystem.  Compress these connected components into single subsystem nodes and replace the input and output nodes as exosystems (instead of signals).  Replace all removed edges following the convention that if a node in $\mathscr{C}$ is no longer in $\mathscr{S}$, connect the edge to the corresponding subsystem node.  This may lead to multiple edges between nodes in $\mathscr{S}$ (e.g. between subsystems), so we compress these edges into a single edge and change the label to be a vector label, reflecting the multiple signals on that edge.

Now, compare the resulting subsystem structure with the results we obtain if we work directly from the equations defining the original subsystems in Example 2 leading up to Equation (\ref{eq:exgssm}).  If we find the transfer function of each subsystem individually, build the associated subsystem matrix $S$, and then interconnect appropriately, we recover the following structured LFT:

\begin{equation}
N = \begin{bmatrix} 0 & 0 & 0 & 0 & 1 & 0 & 0 & 0\\ 0 & 0 & 0 & 0 & 0 & 1 & 0 & 0 \\ 0 & 0 & 0 & 0 & 0 & 0 & 1 & 0 \\ 0 & 0 & 0 & 0 & 0 & 0 & 0 & 1 \\ 1 & 0 & 0 & 0 & 0 & 0 & 0 & 0 \\ 0 & 1 & 0 & 0 & 0 & 0 & 0 & 0 \\ 0 & 0 & 0 & 0 & 0 & 0 & 1 & 0 \\ 0 & 0 & 0 & 0 & 0 & 0 & 0 & 1 \\ 0 & 0 & 1 & 0 & 0 & 0 & 0 & 0 \\ 0 & 0 & 0 & 1 & 0 & 0 & 0 & 0 \\ 0 & 0 & 0 & 0 & 1 & 0 & 0 & 0 \\ 0 & 0 & 0 & 0 & 0 & 1 & 0 & 0 \end{bmatrix} \label{eq:ns}\end{equation} 
\[S = \frac{1}{s^2+3s+2}\left[\begin{array}{cccc|cccc} 0 & s+2 & 0 & s+2 & 0 & 0 & 0 & 0 \\ s+1 & 1 & s+1 & 1 & 0 & 0 & 0 & 0 \\\hline 0 & 0 & 0 & 0 & \frac{2}{s+3} & s+2 & \frac{2}{s+3} & s+2 \\ 0 & 0 & 0 & 0 & \frac{s+1}{s+3} & 0 & \frac{s+1}{s+3} & 0 \end{array}\right]\]

Compare the results of the structured LFT with the signal structure in Figure \ref{fig:substr}.  Notice that building the subsystem structure according to Definition \ref{def:sub} leads to the same result; both processes construct the same graph.  Nevertheless, building subsystem structure directly from $\mathscr{C}$ sheds insight into the meaning of subsystems, as the connected components of the hidden structure of $\mathscr{C}$. 

\end{example}
These procedures uniquely specify $(N,S)$ and $\mathscr{S}$ from a generalized state space model and its complete computational structure, $\mathscr{C}$. This implies that the system models and their associated structural representations considered so far produce a totally ordered set with respect to the relation, ``uniquely specified by."  These are, in order of increasing structural informativity: 
\begin{enumerate}
\item Functional system descriptions and the manifest structure, (which are uniquely specified by)
\item Structured LFTs and the subsystem structure, (which are uniquely specified by) 
\item Generalized state space models and the complete computational structure. 
\end{enumerate}
The next section considers an alternative approach for representing systems, focusing on the interaction among manifest signals as opposed to the interconnection among subsystems.

\subsection{Dynamical Structure Functions and the Signal Structure}
\label{sec:dsf}

One of the difficulties in learning a system's subsystem structure from data is that it necessarily perfectly partitions the system states into subsystem groups, so one must be able to identify the the correct subsystem for each state variable--even those that are ``hidden," or not directly manifest.  This section considers a system representation that precisely characterizes the interaction between manifest signals without drawing any conclusions about ``hidden" variables.  This ability to remain agnostic about the structural role of hidden variables not only makes this representation easier to learn from data, but it also makes it extremely useful for describing systems with a ``fluidic" component that makes the very idea of subsystems difficult to conceptualize, such as chemical reaction processes or market behavior.  

This representation, called the {\em dynamical structure function} (DSF), like the structured LFT and the subsystem structure, is part of a totally ordered set with respect to the relation, ``uniquely specified by."  This is, in order of increasing structural informativity: 
\begin{enumerate}
\item Functional system descriptions and the manifest structure, (which are uniquely specified by)
\item Dynamical structure functions and the signal structure, (which are uniquely specified by) 
\item Zero-intricacy state space models and their associated computational structure, (which are uniquely specified by)
\item Generalized state space models and the complete computational structure.
\end{enumerate}
Note that the structured LFT and the subsystem structure are not listed as part of this ordering.  This is because, although the subsystem structure falls between the manifest and complete computational structures (as described in the previous section), its relationship to signal structure and the zero intricacy realization is ambiguous in general, depending on each case individually; more on this will be discussed later.  Likewise, it is interesting to note that the DSF naturally scales between functional system descriptions and zero-intricacy state space models depending on the number of independent measured outputs.  That is to say, the DSF of a single-output system is equivalent to its functional description, while the DSF of the same system, except with full state measurements, is equivalent to its state space model.  These ideas will be made precise next. 

Define a system's DSF by considering the zero-intricacy realization, $(A,B,C,D)$, of a generalized state space model, as in (\ref{eq:ssequation}); this is the standard state space model generally considered in the literature.  Auxiliary variables to characterize the intricacy of functional composition in a generalized state space model, as in (\ref{eq:genss}), are only necessary for specifying subsystem structure; they play no role in DSF theory.  Without loss of generality, let $p_1\leq p$ be the rank of $C$ and assume it has the form:     
\begin{equation}
C = \begin{bmatrix} C_{11} & C_{12} \\ C_{21} & C_{22} \end{bmatrix}
\label{eq:C}
\end{equation}
where $C_{11}\in\mathbb{R}^{p_1 \times p_1}$ is invertible.  Note that any system can be put into this form with a simple renumbering of the output signals and the state variables.  

Let $E$ be any basis of the null space of $C$, and partition $E=\begin{bmatrix}E_1&E_2\end{bmatrix}^T$ commensurate with the partitioning of $C$ in (\ref{eq:C}).  Note that $E_2$ is necessarily square, since $E$ is size $n\times (n-p_1)$ and $E_1$ is $p_1\times (n-p_1)$, implying $E_2$ has dimensions $(n-p_1)\times(n-p_1)$.  Moreover, $E_2$ is necessarily invertible.  This is seen by contradiction: suppose $E_2$ is not invertible.  Then there is a vector $z\in\mathbb{R}^{(n-p_1)}\neq 0$ such that $E_2z=0$.  This would mean, however, that $Ez=0$, since $C_{11}E_1+C_{12}E_2=0$ and $C_{11}$ invertible together imply that $E_1 = -C_{11}^TC_{12}E_2$, but $Ez=0$ is a contradiction because $E$ is, by definition, a collection of linearly independent vectors forming a basis for the null space of $C$.

Now, consider a state transformation on (\ref{eq:ssequation}) of the form $z = Tx$, where
\begin{equation}
T = \begin{bmatrix} C_{11}^{-1} & E_1 \\ 0 & E_2 \end{bmatrix},\;\;\; \text{ and }\;\;\; T^{-1} = \begin{bmatrix}C_{11} & C_{12} \\ 0 & E_{2}^{-1}\end{bmatrix}.\label{eq:transform}\end{equation}
This state transformation yields a system of the form:
\begin{equation}\begin{array}{rcl}\begin{bmatrix}\dot{z}_1 \\ \dot{z}_2\end{bmatrix} & = & \begin{bmatrix} A_{11} & A_{12} \\ A_{21} & A_{22}\end{bmatrix}\begin{bmatrix}z_1 \\ z_2\end{bmatrix}+ \begin{bmatrix}B_1 \\ B_2\end{bmatrix}u \\ \begin{bmatrix}y_1 \\ y_2\end{bmatrix} & = & \begin{bmatrix}I & 0 \\ C_{21}C_{11}^{-1}& 0\end{bmatrix}\begin{bmatrix}z_1 \\ z_2\end{bmatrix} + \begin{bmatrix}D_1 \\ D_2\end{bmatrix}u \end{array}\label{eq:abcd}\end{equation}
where $z_{1}\in\mathbb{R}^{p_1}$, $z_{2}\in\mathbb{R}^{n-p_1}$, $y_{1}\in\mathbb{R}^{p_1}$, $y_2\in\mathbb{R}^{p-p_1}$ and $u\in\mathbb{R}^{m}$.
To avoid unnecessary notation, we will refer to the state matrices in (\ref{eq:abcd}) as $(A,B,C,D)$ from (\ref{eq:abcd}), as opposed to the equivalent but distinct matrices $(A,B,C,D)$ from (\ref{eq:ssequation}); we don't expect this slight abuse of notation to cause any confusion.  By way of comparison between these two realizations, however, a few comments may be in order.  First, note that the transformation resulting in ($\ref{eq:abcd})$ has redefined the system in terms of the first $p_1$ signals in $y$, $y_1$, so that $z_1$ are manifest states (once $D_1u$ has been considered) and $z_2$ are hidden states.  This suggests that the computational structure of $A$ from (\ref{eq:abcd}), in some sense, describes how the manifest states affect each other, both directly and indirectly through the hidden states, while the original description of $A$ from (\ref{eq:ssequation}) scrambled this information through the change of basis, $T$.  Next, note that although the manifest states $z_1$ are clearly observable, $z_2$ being ``hidden" does not necessarily imply that they are unobservable--just that they must be observed through their influence on $z_1$.  In general we make no assumptions about the observability or controllability of the system, and the structure of (\ref{eq:abcd}) only ensures that any unobservable states must belong to $z_2$, but not that every state in $z_2$ is unobservable.  Finally, notice that the additional $p-p_1$ outputs in $y_2$ are redundant, merely ``snapping" onto the rest of the system and playing no role whatsoever in how information flows from $u$ to $y_1$.    

Finding the DSF from (\ref{eq:abcd}) is now straight forward.  In \cite{dsfpaper} the DSF was defined for systems of the form $(A,B,[\begin{array}{cc}I&0\end{array}], 0)$; we follow that derivation here and extend it for the additional structure in (\ref{eq:abcd}).  Taking Laplace transformations of the state equation in (\ref{eq:abcd}) yields:   

\begin{equation}\begin{bmatrix}sZ_1 \\ sZ_2\end{bmatrix} = \begin{bmatrix}sA_{11} & A_{12} \\ A_{21} & A_{22}\end{bmatrix}\begin{bmatrix} Z_1 \\ Z_2\end{bmatrix} + \begin{bmatrix}B_1 \\ B_2\end{bmatrix}U \label{eq:ztran}\end{equation}
where $Z(s)$ is the Laplace transform of $z(t)$ and $U(s)$ is the transform of $u(t)$.  Solving for $Z_2$ in the second equation and substituting it into the first equation in (\ref{eq:ztran}) then yields:
\begin{equation}
sZ_1 = WZ_1+VU
\label{eq:WV}
\end{equation}
where $W = A_{11}+A_{12}(sI-A_{22})^{-1}A_{21}$ and $V = B_1 + A_{12}(sI-A_{22})^{-1}B_2$. Let $D_W$ to be the matrix of diagonal entries of $W$, and subtract $D_WZ_1$ from both sides of (\ref{eq:WV}), yielding \begin{equation}Z_1=QZ_1+PU\label{eq:z}\end{equation} where $Q = (sI-D)^{-1}(W-D_W)$ and $P=(sI-D_W)^{-1}V$.  These matrix functions of the Laplace variable, $(Q(s),P(s))$ would be the DSF of the system with output structure $y=[\begin{array}{cc}I&0\end{array}]z$, and it is relatively easy to see that $Q$ and $P$ have certain properties, such as being strictly proper rational functions, or that the diagonal entries of $Q$ are identically zero.  We extend this definition of the DSF by noting from (\ref{eq:abcd}) that $Z_1=Y_1-D_1U$.  Substituting into (\ref{eq:z}) yields
\begin{equation}\begin{bmatrix}Y_1 \\ Y_2\end{bmatrix}=\bar{Q}Y_1+\bar{P}U\end{equation} where $\bar{Q} = \begin{bmatrix}Q \\ C_{21}C_{11}^{-1}\end{bmatrix}$, $\bar{P} = \begin{bmatrix}P + (I-Q)D_1 \\ D_2 - C_{21}C_{11}^{-1}D_1\end{bmatrix}$. Note that when $C = \begin{bmatrix}I & 0\end{bmatrix}$ and $D = 0$, $\bar{Q} = Q$ and $\bar{P} = P$, so there should be no confusion referring to either $(\bar{Q},\bar{P})$ or $(Q,P)$ as the dynamical structure function, since $(\bar{Q},\bar{P})$ simply extends the previous definition to the general case.  

\begin{theorem}
Given a zero-intricacy state space model as in (\ref{eq:ssequation}) with $C$ of the form given in (\ref{eq:C}) with $C_{11}$ invertible, the dynamical structure function $(Q,P)$ is uniquely specified.
\end{theorem}
\begin{proof}
We will show that although a transformation of the form given in (\ref{eq:transform}) results in different state equations, they still produce the same dynamical structure function. Given that $E_1 = -C_{11}^TC_{12}E_2$, the transformation resulting from any particular choice of $E_2$ is given by
\[T = \begin{bmatrix}C_{11}^{-1} & -C_{11}^TC_{12}E_2 \\ 0 & E_2 \end{bmatrix}\]
which results in a state space transformation given in (\ref{eq:abcd}), where $A$ and $B$ have the form:

\begin{equation}
\begin{array}{rcl}
\begin{bmatrix}A_{11}&A_{12}\\A_{21}&A_{22}\end{bmatrix} & = & \begin{bmatrix}C_{11} & C_{12} \\ 0 & E_{2}^{-1}\end{bmatrix}\begin{bmatrix}\hat{A}_{11} & \hat{A}_{12} \\ \hat{A}_{21} & \hat{A}_{22}\end{bmatrix}\begin{bmatrix}C_{11}^{-1} & -C_{11}^TC_{12}E_2 \\ 0 & E_2 \end{bmatrix}, \\ 
\implies A_{11} & = & (C_{11}\hat{A}_{11}+C_{12}\hat{A}_{21})C_{11}^{-1} \\ 
\implies A_{12} & = & (C_{11}(\hat{A}_{12}-\hat{A}_{11}C_{11}^{-1}C_{12}) + C_{12}(\hat{A}_{22}-\hat{A}_{21}C_{11}^{-1}C_{12}))E_2 \\ 
\implies A_{21} & = & E_{2}^{-1}(\hat{A}_{21}C_{11}^{-1}) \\ 
\implies A_{22} & = & E_{2}^{-1}(\hat{A}_{22}-\hat{A}_{21}C_{11}^{-1}C_{12})E_{2}, \\ 
\begin{bmatrix}B_{1}\\ B_{2}\end{bmatrix} & = & \begin{bmatrix}C_{11} & C_{12} \\ 0 & E_{2}^{-1}\end{bmatrix}\begin{bmatrix}\hat{B}_{1}\\ \hat{B}_{2}\end{bmatrix} = \begin{bmatrix}C_{11}\hat{B}_1+C_{12}\hat{B}_2 \\ E_{2}^{-1}\hat{B}_2\end{bmatrix}
\end{array}\label{eq:abtrans}\end{equation}
where $\hat{A}$ and $\hat{B}$ are the untransformed zero intricacy state matrices.

Clearly, different choices of $E_2$ can lead to considerably different state matrices in (\ref{eq:abtrans}), we will now show that for different choices of $E_2$ the dynamical structure function does not change. From (\ref{eq:WV}), we know that
\begin{equation}\begin{array}{rcl}W & = & A_{11}+A_{12}(sI-A_{22})^{-1}A_{21}, \\  V & = & B_1 + A_{12}(sI-A_{22})^{-1}B_2\end{array}\end{equation}
which by direct substitution is invariant to perturbations in $E_2$. Invariance of $W$ and $V$ imply invariance of $Q$ and $P$, which completes the proof.
\end{proof}

The graphical representation of the dynamical structure function is known as the \emph{signal structure} of a system and is denoted $\mathscr{W}$, with a vertex set $V(\mathscr{W})$ and edge set $E(\mathscr{W})$, \cite{structure}. The elements of a system's signal structure is defined to be:

\begin{itemize}
\item $V(\mathscr{W}) = \{u_{1}, ..., u_{m}, y_{11}, ..., y_{1p_1}, y_{21}, ..., y_{2p_2}\}$, each representing a manifest variable of the system with $p_2=p-p_1$, and
\item $E(\mathscr{W})$ contains an edge from $v_i \in V(\mathscr{W})$ to $v_j \in V(\mathscr{W})$ if the associated entry of $\bar{Q}$ and $\bar{P}$ is nonzero.
\end{itemize}

Unlike the subsystem structure, the signal structure uses circular nodes to denote signals rather than systems, while the edges between these signals represent systems since it is a condensation graph of the signal flow representation of the complete computational structure.
\begin{example}

Given the generalized state space model in (\ref{eq:exgssm}) with complete computational structure in Figure \ref{fig:genccs}, the dynamical structure function is given in (\ref{eq:dsfqp}), the procedure for determining the corresponding signal structure from a system's generalized state space model is then outlined in Figure \ref{fig:developsigstructure}.

\begin{equation} Q = \begin{bmatrix} 0 & 0 & 0 & \frac{1}{s+1} \\ \frac{1}{s+2} & 0 & \frac{1}{s+2} & 0 \\ 0 & \frac{1}{s+1} & 0 & \frac{2}{s+1} \\ \frac{1}{s^2+5s+6} & 0 & 0 & 0 \end{bmatrix} \text{,   } P = \begin{bmatrix} 0 & \frac{1}{s+1} & 0 & 0 \\ \frac{1}{s+2} & 0 & 0 & 0 \\ 0 & 0 & 0 & \frac{1}{s+1} \\ 0 & 0 & \frac{1}{s^2+5s+6} & 0\end{bmatrix} \label{eq:dsfqp}\end{equation}

\begin{figure}[h!]
        \centering
        \begin{subfigure}[b]{0.4\textwidth}
                \includegraphics[page=1,width=\textwidth]{examples}
                \caption{Step 1: Begin with a generalized Complete Computational Structure.}
                \label{fig:bgenccs}
        \end{subfigure}
        \begin{subfigure}[b]{0.4\textwidth}
                \includegraphics[page=8,width=\textwidth]{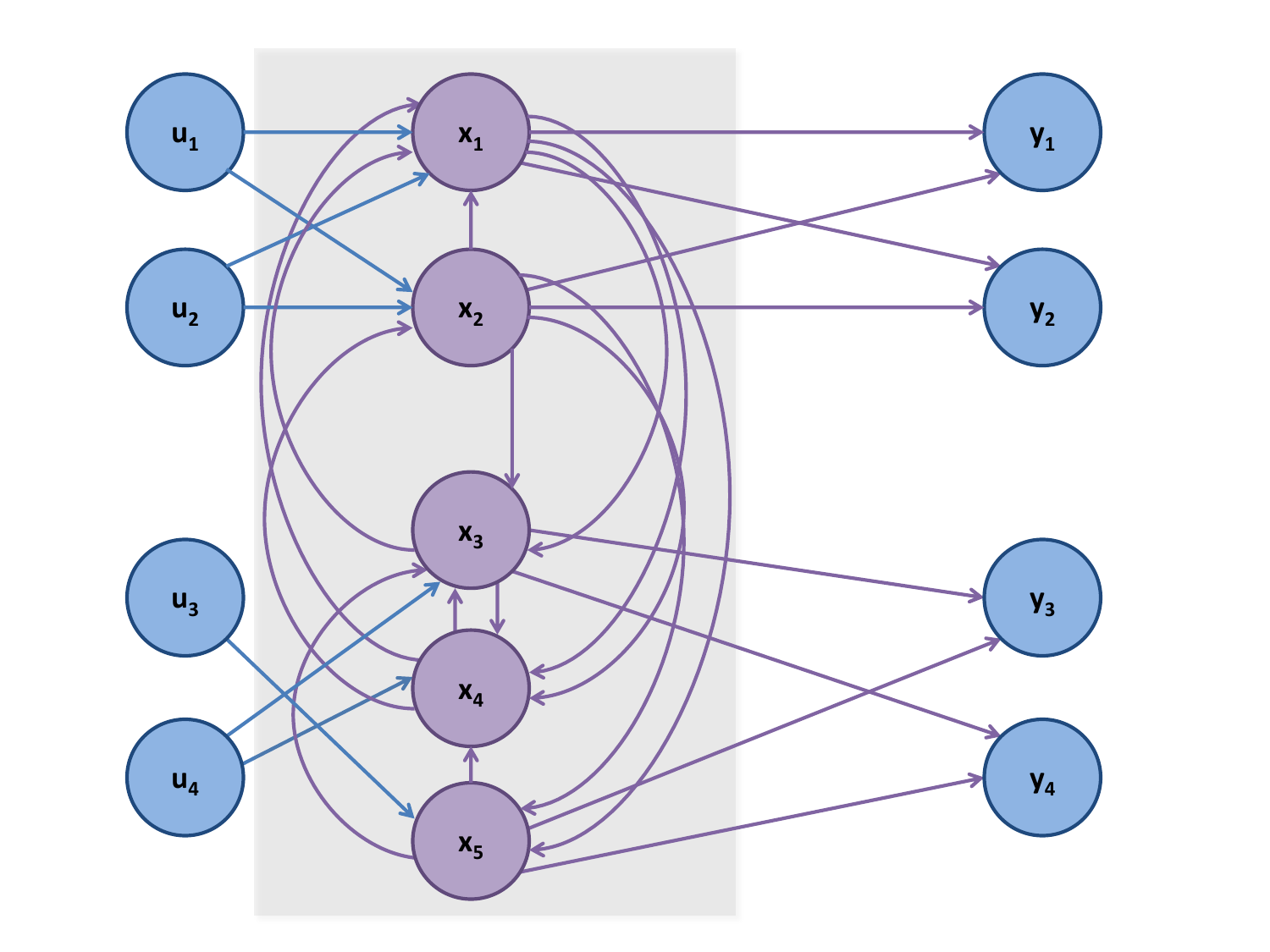}
                \caption{Step 2: Find the corresponding zero intricacy Complete Computational Structure.}
                \label{fig:bziccs}
        \end{subfigure}
        \begin{subfigure}[b]{0.4\textwidth}
                \includegraphics[page=9,width=\textwidth]{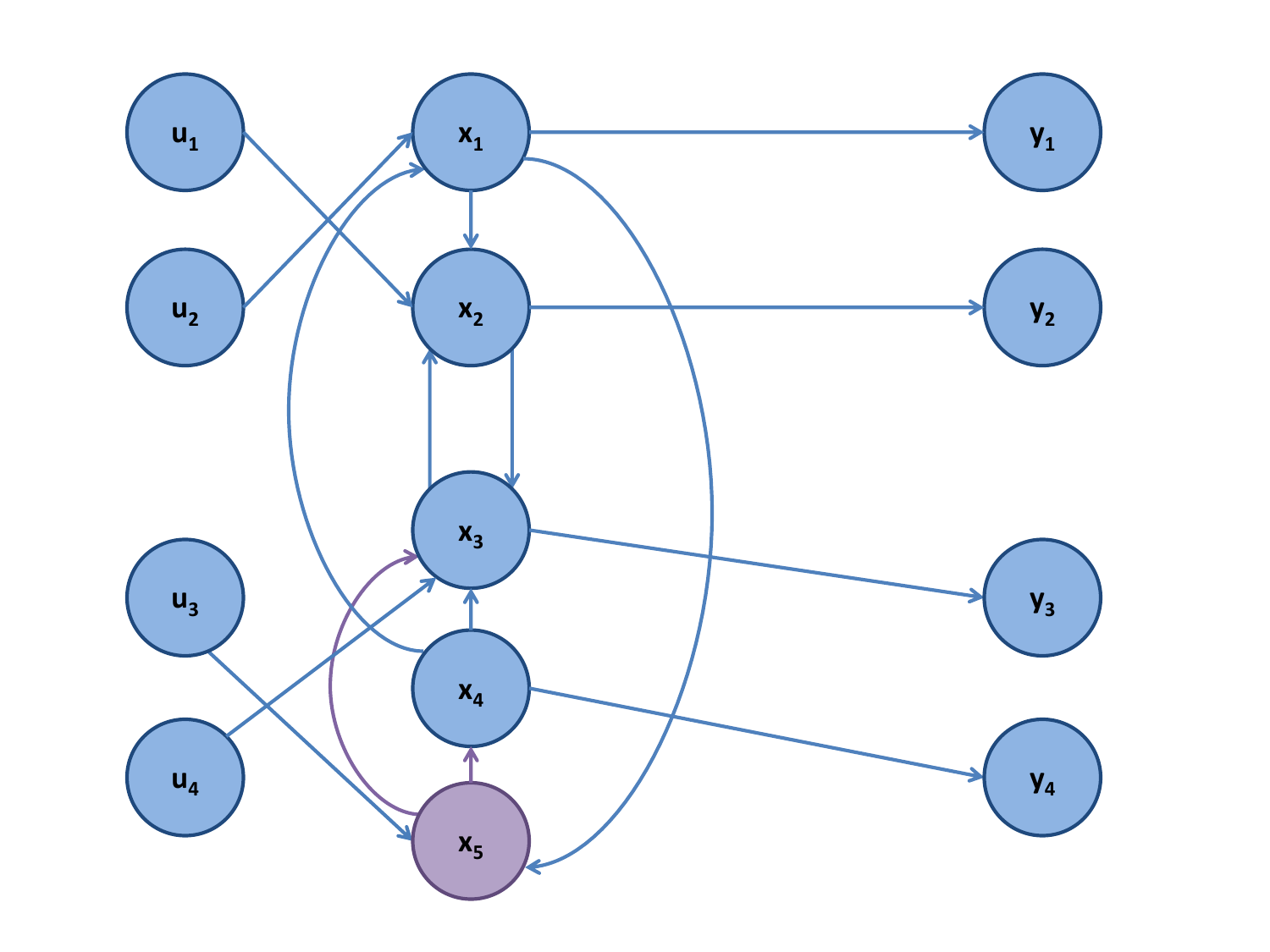}
                \caption{Step 3: Transform the system so that $C= \begin{bmatrix}I & 0\end{bmatrix}$. If $C_{11}$ is invertible, then the transformed system does not change the structure of the signal structure.}
                \label{fig:transccs}
        \end{subfigure}
        \begin{subfigure}[b]{0.4\textwidth}
                \includegraphics[page=10,width=\textwidth]{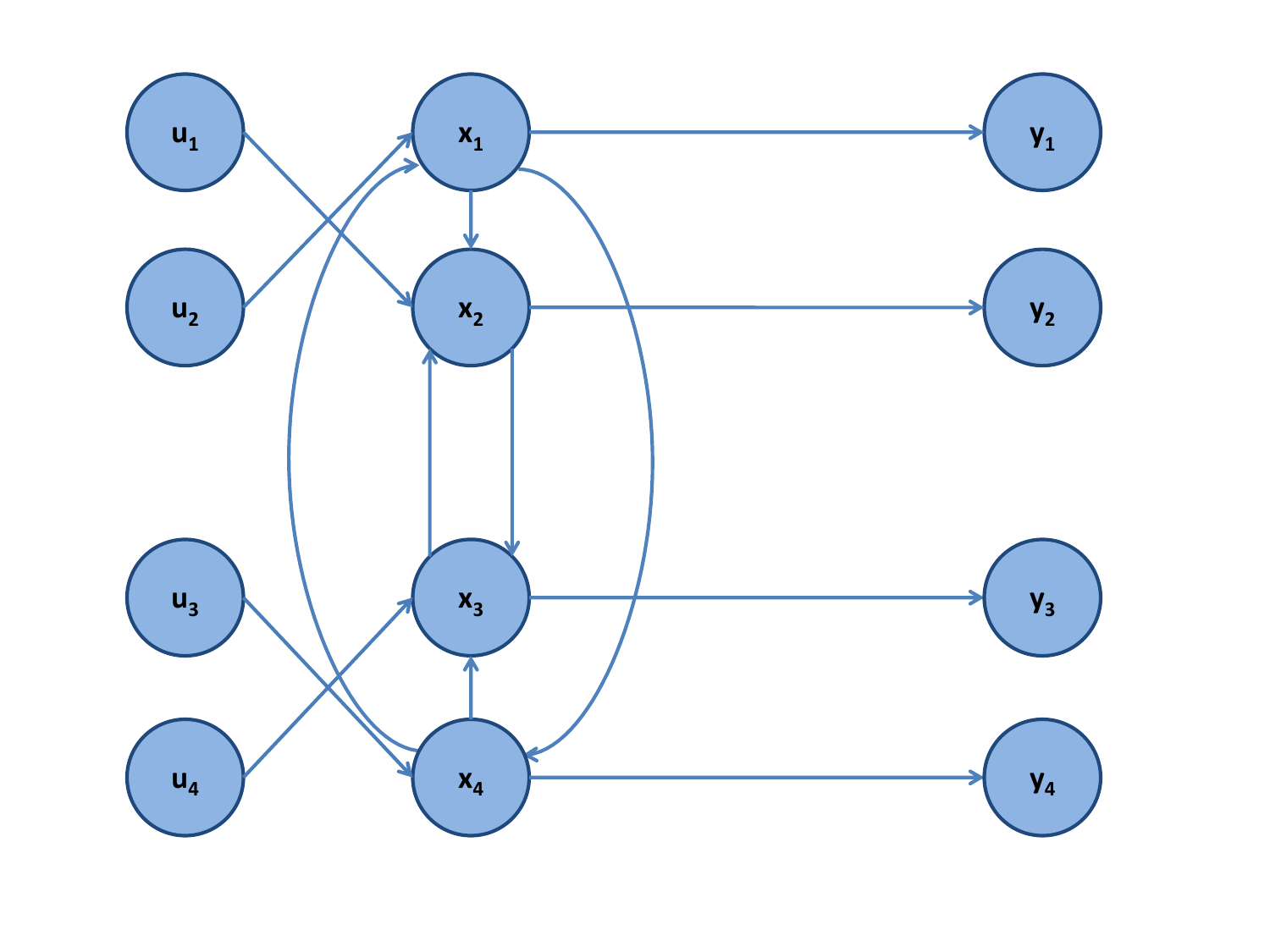}
                \caption{Step 4: Remove hidden nodes from the system, maintaining pathways from manifest variable $i$ to manifest variable $j$ through the removed hidden nodes.}
                \label{fig:remhidcss}
        \end{subfigure}
        \begin{subfigure}[b]{0.5\textwidth}
                \includegraphics[page=11,width=\textwidth]{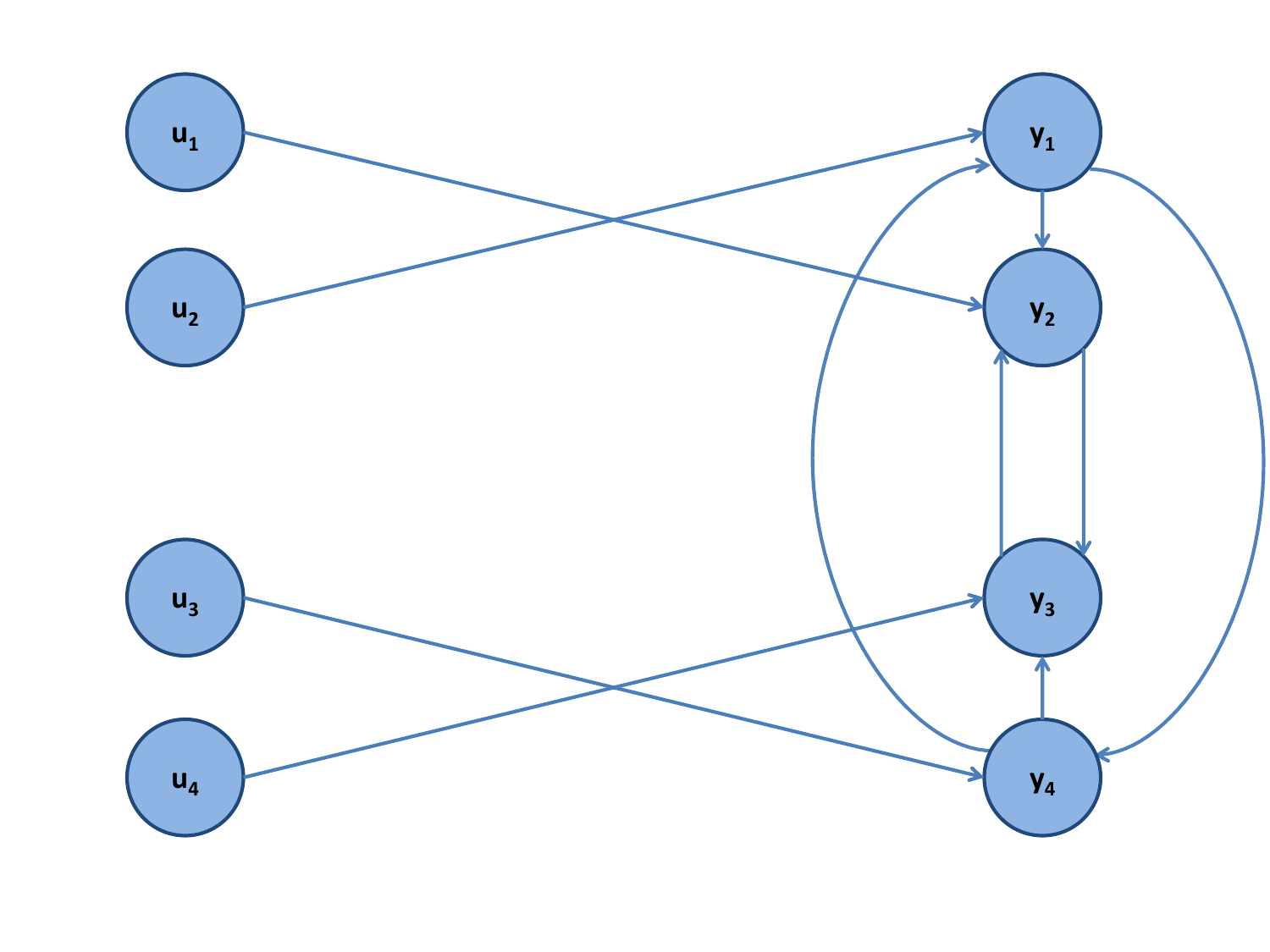}
                \caption{Step 5: Rename manifest state variables, $x$, to their corresponding output variables, $y$, while removing the edges from manifest states to outputs along with the corresponding node.}
                \label{fig:sigstr}
        \end{subfigure}
        \caption{Signal structure, built from a system's complete computational structure}\label{fig:developsigstructure}
\end{figure}
\end{example}

\subsection{Relationship to Other System Representations}

The transfer function associated with a given dynamical structure function is given by \begin{equation}H(s) = (I-Q)^{-1}P\label{eq:tfdsf1}\end{equation} which is found easily from (\ref{eq:z}). We note that $I-Q$ is invertible since $Q$ is a square, hollow transfer function matrix, so $I-Q$ will always have full rank. Necessary and sufficient conditions for determing a dynamical structure function given a system's transfer function were developed in \cite{nessandsuff}. This process is known as \emph{network reconstruction} and is discussed in detail in Section \ref{sec:netrecon}.

Comparing the signal structure to subsystem structure, Example \ref{ex:twosub} shows that it is possible for a system's signal structure to be consistent with two or more subsystem structure representations. Example \ref{ex:twosig} shows it is also possible for a system's subsystem structure to be consistent with two or more dynamical structure functions. The implication of this result is that these two partial structure system representations denote two different notions of structure within a system.

\begin{example}
Given the complete computational structure shown in Figure \ref{fig:ssccs1}, the associated subsystem structure was found to be that shown in Figure \ref{fig:ss1}. 

Given a second complete computational structure in Figure \ref{fig:ssccs2}, the associated subsystem structure is shown in Figure \ref{fig:ss2}. Note that this complete computational structure is the same structure as the computational structure of the zero intricacy realization of Figure \ref{fig:ssccs1}. The two are distinguished by the fact that the zero intricacy structure in Figure \ref{fig:ziccs} is a computational structure, meaning that it is not complete and the associated complete computational structure required auxiliary variables to model various compositions of functions. The complete computational structure given in Figure \ref{fig:ssccs2}, however, has the same structure, but no auxiliary variables were utilized for composition of functions, so the structure is considered complete.

The corresponding signal structure for both complete computational structures is then given in Figure \ref{fig:samesig}, thus we have shown that it is possible for a single signal structure to be consistent with multiple subsystem structures.

\label{ex:twosub}
\begin{figure}[h!]
        \centering
        \begin{subfigure}[b]{0.45\textwidth}
                \includegraphics[page=1,width=\textwidth]{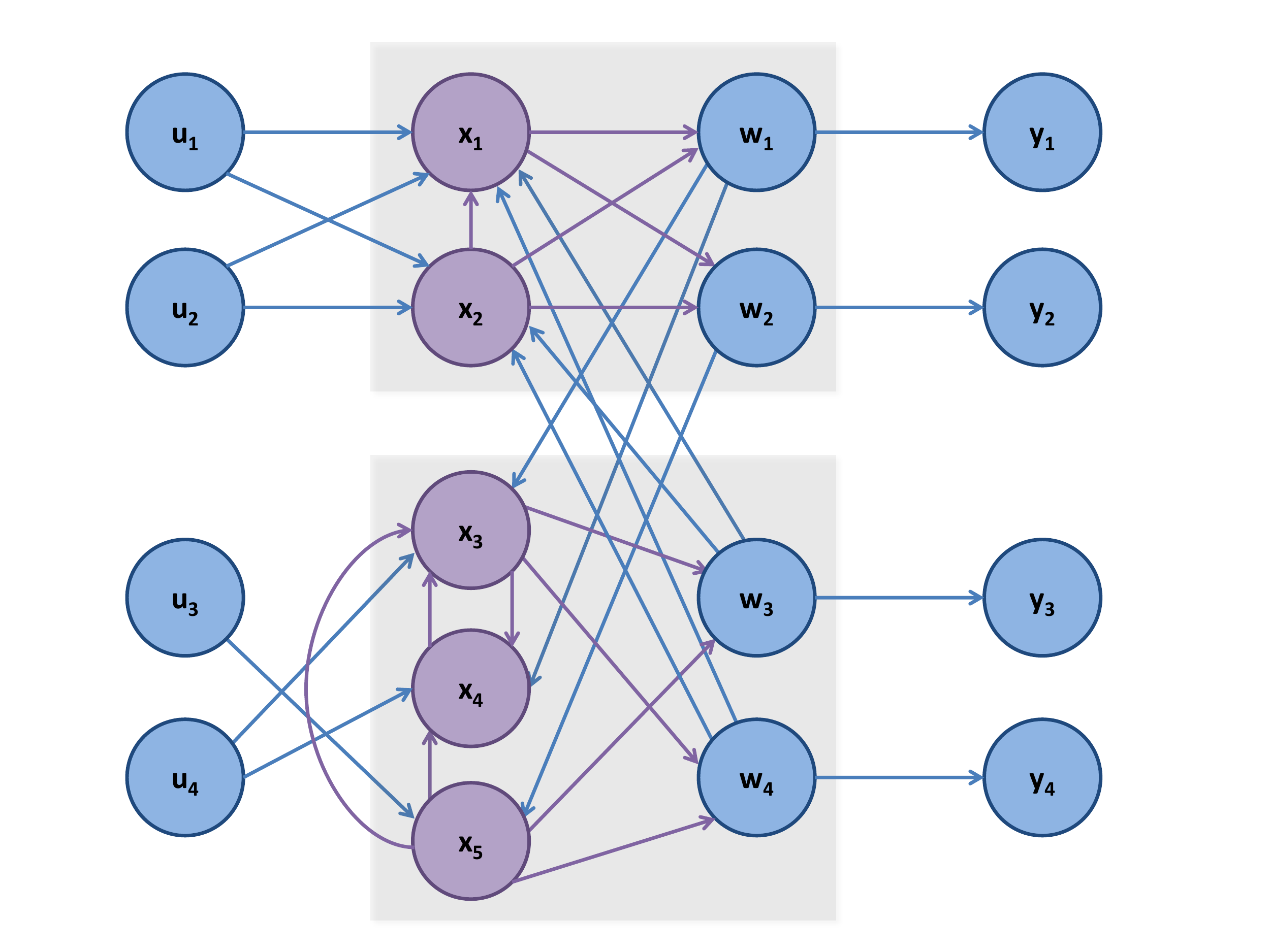}
                \caption{The complete computational structure of a system.}
                \label{fig:ssccs1}
        \end{subfigure}
        \begin{subfigure}[b]{0.45\textwidth}
                \includegraphics[page=2,width=\textwidth]{dsfsubexample2}
                \caption{This yields a subsystem structure of two systems in feedback.}
                \label{fig:ss1}
        \end{subfigure}
        \begin{subfigure}[b]{0.45\textwidth}
                \includegraphics[page=3,width=\textwidth]{dsfsubexample2}
                \caption{A complete computational structure without any intricacy variables.}
                \label{fig:ssccs2}
        \end{subfigure}
        \begin{subfigure}[b]{0.45\textwidth}
                \includegraphics[page=4,width=\textwidth]{dsfsubexample2}
                \caption{This yields a new subsystem structure with a single subsystem.}
                \label{fig:ss2}
        \end{subfigure}
        \begin{subfigure}[b]{0.55\textwidth}
                \includegraphics[page=5,width=\textwidth]{dsfsubexample2}
                \caption{The signal structure is the same for both complete computational structures.}
                \label{fig:samesig}
        \end{subfigure}
        \caption{Signal Structure consistent with Two Subsystem Structures}\label{fig:samestructure}
\end{figure}
\end{example}

\begin{example}
Given the complete computational structure in Figure \ref{fig:remedge}, which is the complete computational structure from Figure \ref{fig:genccs} with an edge removed (highlighted in red), the associated subsystem structure (shown in Figure \ref{fig:subsame}) does not change.

The associated computational structure of Figure \ref{fig:remedge} found by determining the zero intricacy state space model is given in Figure \ref{fig:ziremedges} and is similar to the computational structure given in Figure \ref{fig:ziccs} although with an edge missing (again, marked in red).

Transforming the system to get $C = \begin{bmatrix} I & 0 \end{bmatrix}$, yields the structure given in Figure \ref{fig:zitransmoreedges}, which is similar to the transformed structure given in Figure \ref{fig:transccs}, with an extra edge, highlighted in red. The associated signal structure is then given in Figure \ref{fig:dsfdiff}, also containing an extra edge meaning the subsystem structure given in Figure \ref{fig:subsame} is consistent with multiple signal structures.

\label{ex:twosig}
\begin{figure}[h!]
        \centering
        \begin{subfigure}[b]{0.45\textwidth}
                \includegraphics[page=2,width=\textwidth]{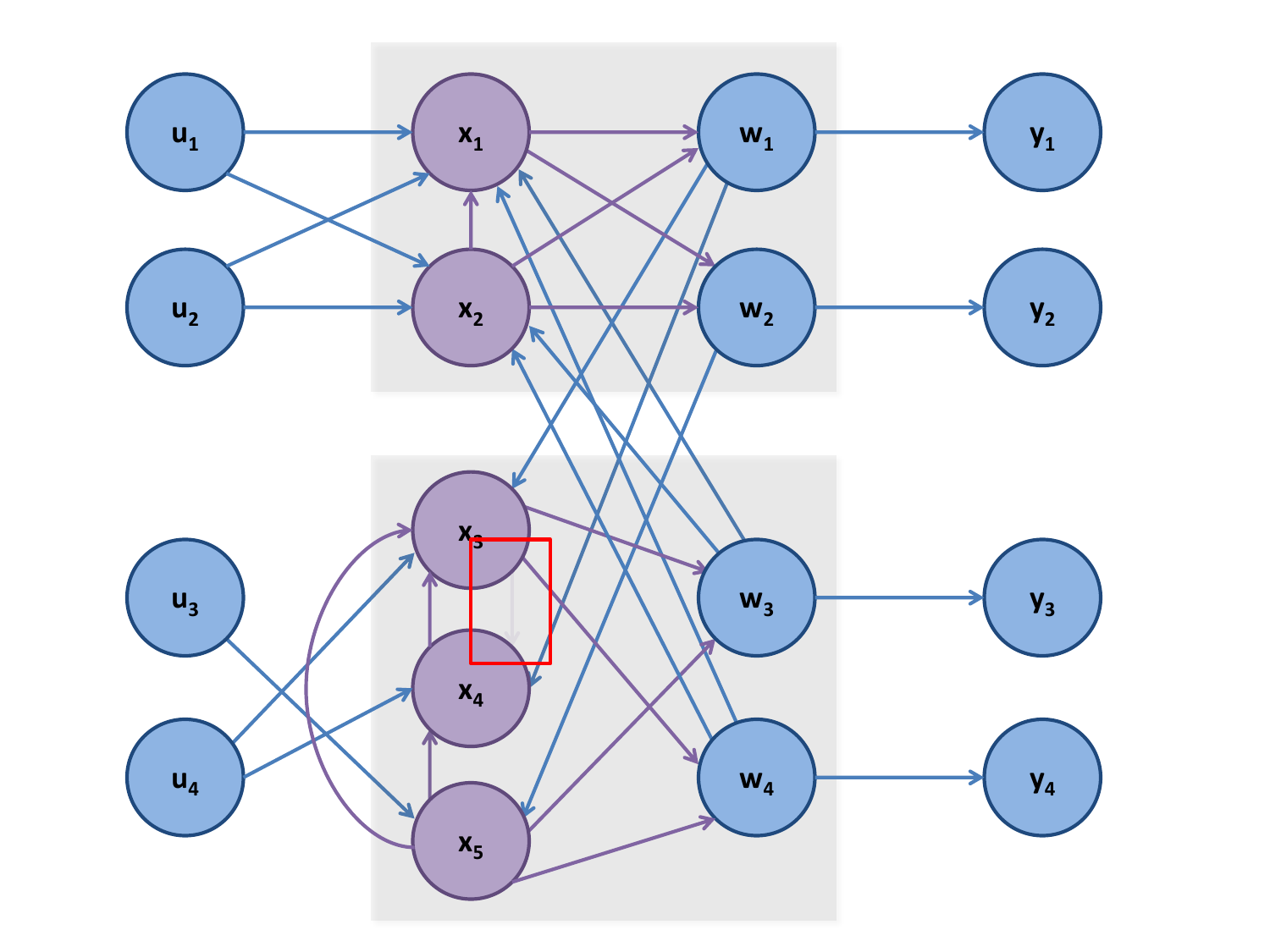}
                \caption{Remove an edge from the complete computational structure.}
                \label{fig:remedge}
        \end{subfigure}
        \begin{subfigure}[b]{0.45\textwidth}
                \includegraphics[page=3,width=\textwidth]{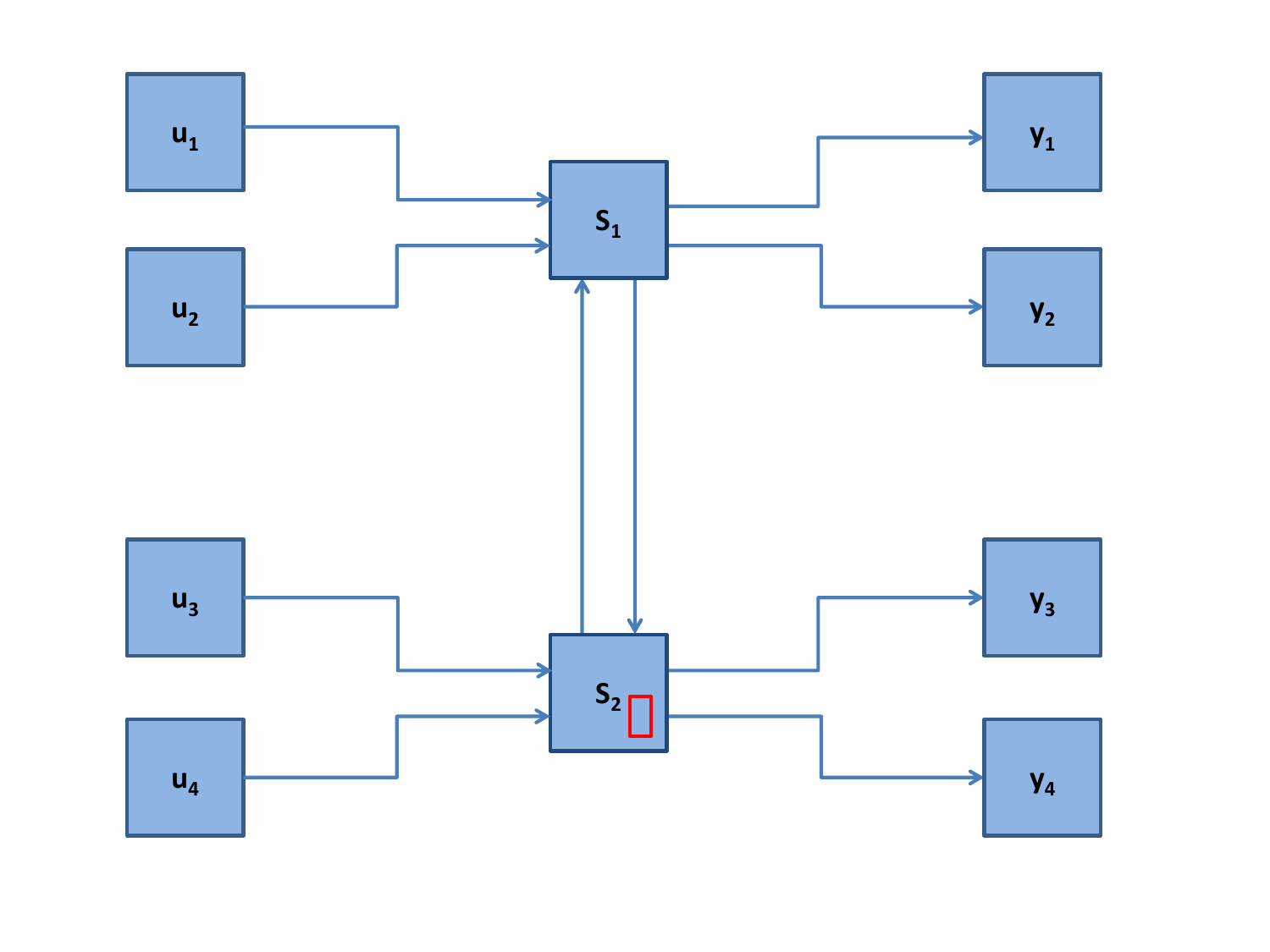}
                \caption{The subsystem structure remains the same since, with two subsystems in feedback.}
                \label{fig:subsame}
        \end{subfigure}
        \begin{subfigure}[b]{0.45\textwidth}
                \includegraphics[page=4,width=\textwidth]{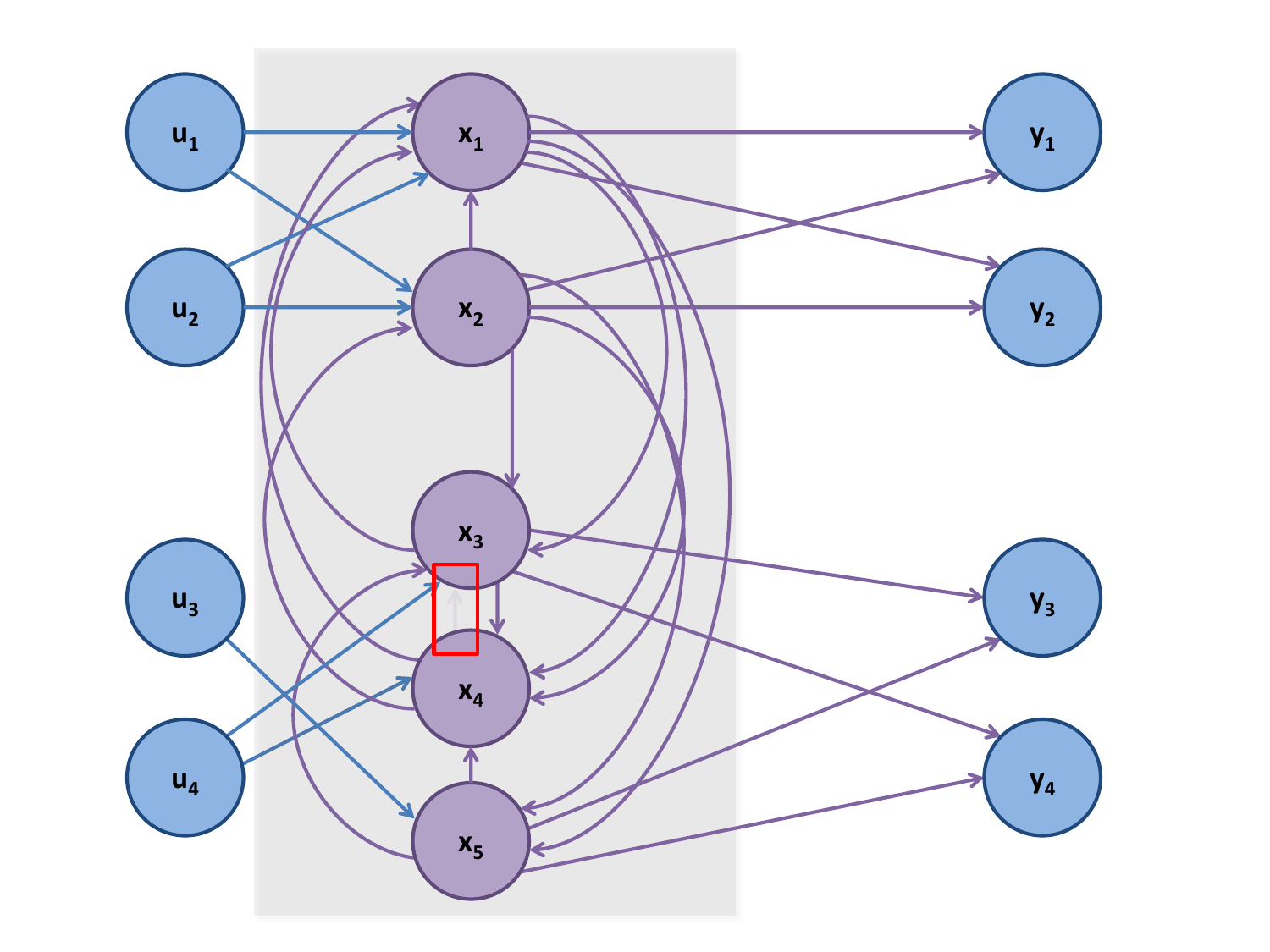}
                \caption{The zero intricacy complete computational structure with the removed edge.}
                \label{fig:ziremedges}
        \end{subfigure}
        \begin{subfigure}[b]{0.45\textwidth}
                \includegraphics[page=5,width=\textwidth]{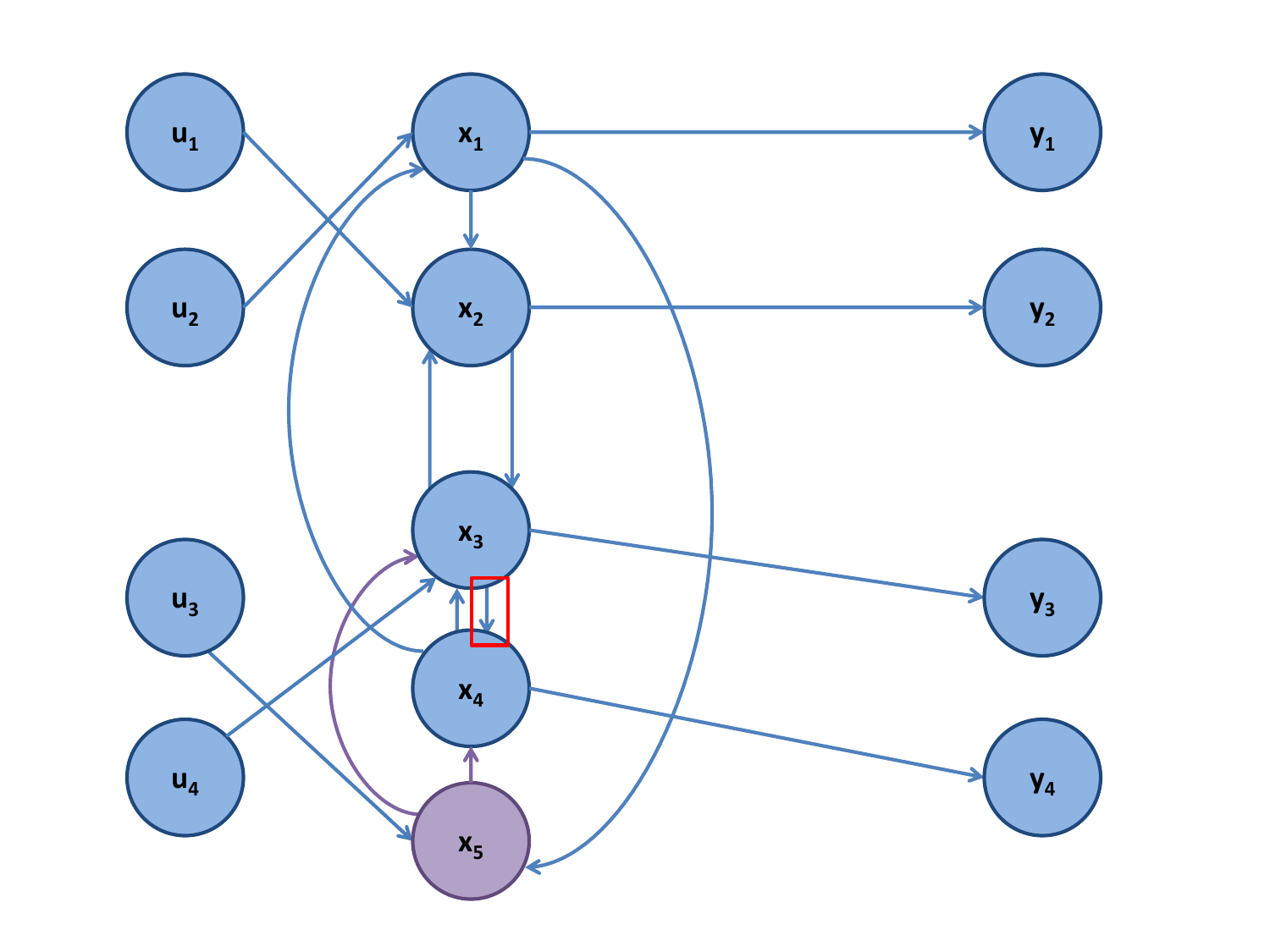}
                \caption{The transformed zero intricacy complete computational structure with an extra edge.}
                \label{fig:zitransmoreedges}
        \end{subfigure}
        \begin{subfigure}[b]{0.55\textwidth}
                \includegraphics[page=6,width=\textwidth]{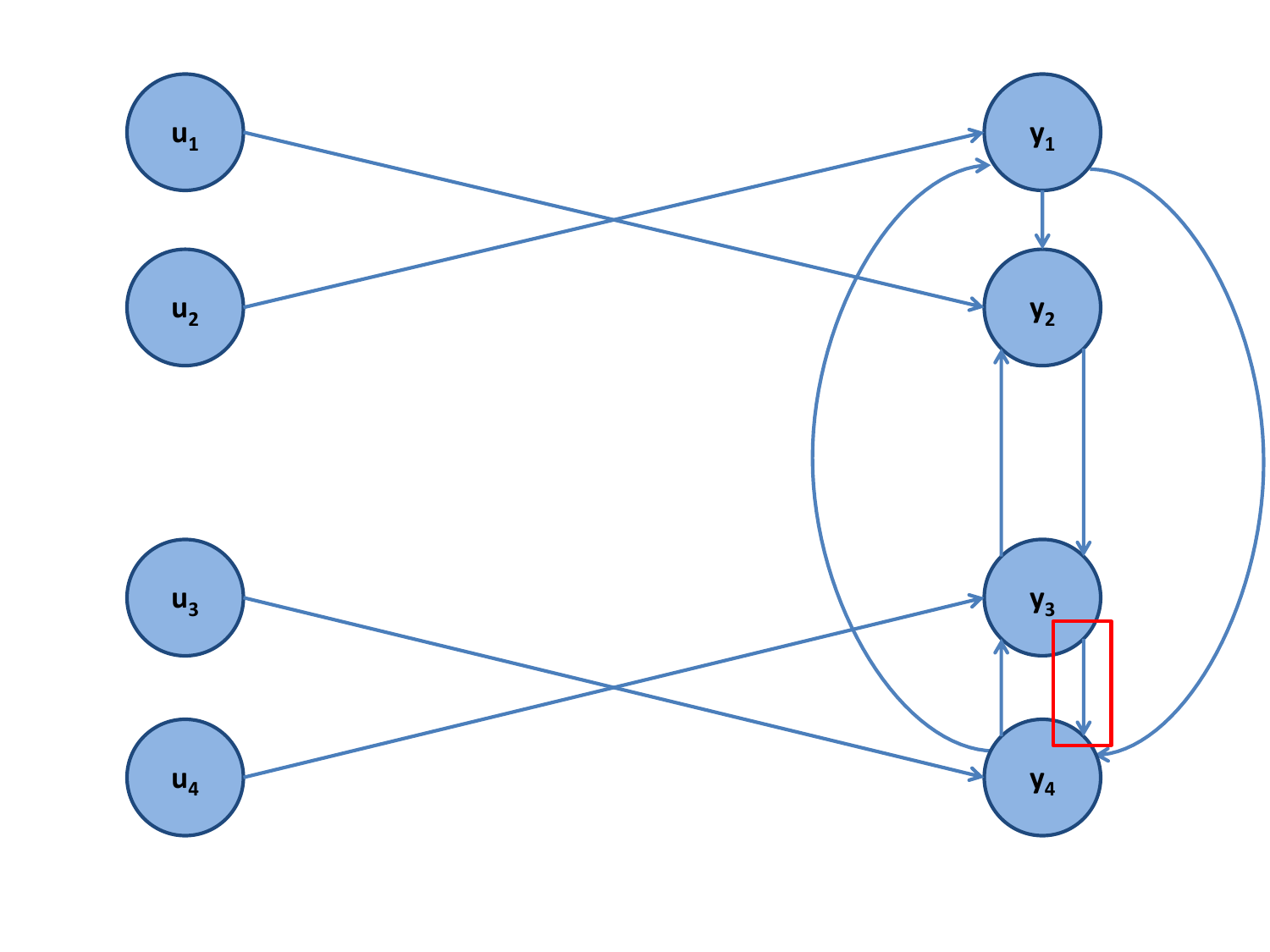}
                \caption{The new signal structure with an extra edge.}
                \label{fig:dsfdiff}
        \end{subfigure}
        \caption{Subsystem Structure consistent with Two Signal Structures}\label{fig:samestructure2}
\end{figure}
\end{example}

One of the properties of the signal structure of a system that distinguishes it from the subsystem structure is known as shared hidden states.

\begin{definition}
A shared hidden state is a state within a system that is not manifest, i.e. that is part of the hidden structure, that has either multiple pathways from it that lead towards a manifest structure or multiple pathways that come from manifest structure or both.    
\end{definition}

When a system contains a shared hidden state, the associated signal structure is agnostic to that state, meaning that it allows for hidden states to be shared across system edges. In contrast, the subsystem structure does not allow for hidden states to be shared across systems. 

Therefore, when shared hidden states exist in a system the signal structure contains more structural information than the subsystem structure as shown in Figure \ref{fig:sharedhiddenstate}. Moreover, since the signal structure is agnostic to shared hidden states, the process of determining a unique dynamical structure function from a system's transfer function has reasonable conditions, see Section \ref{sec:netrecon}, unlike the subsystem structure of the system.

\begin{figure}[h!]
        \centering
        \begin{subfigure}[b]{0.4\textwidth}
                \includegraphics[page=1,width=\textwidth]{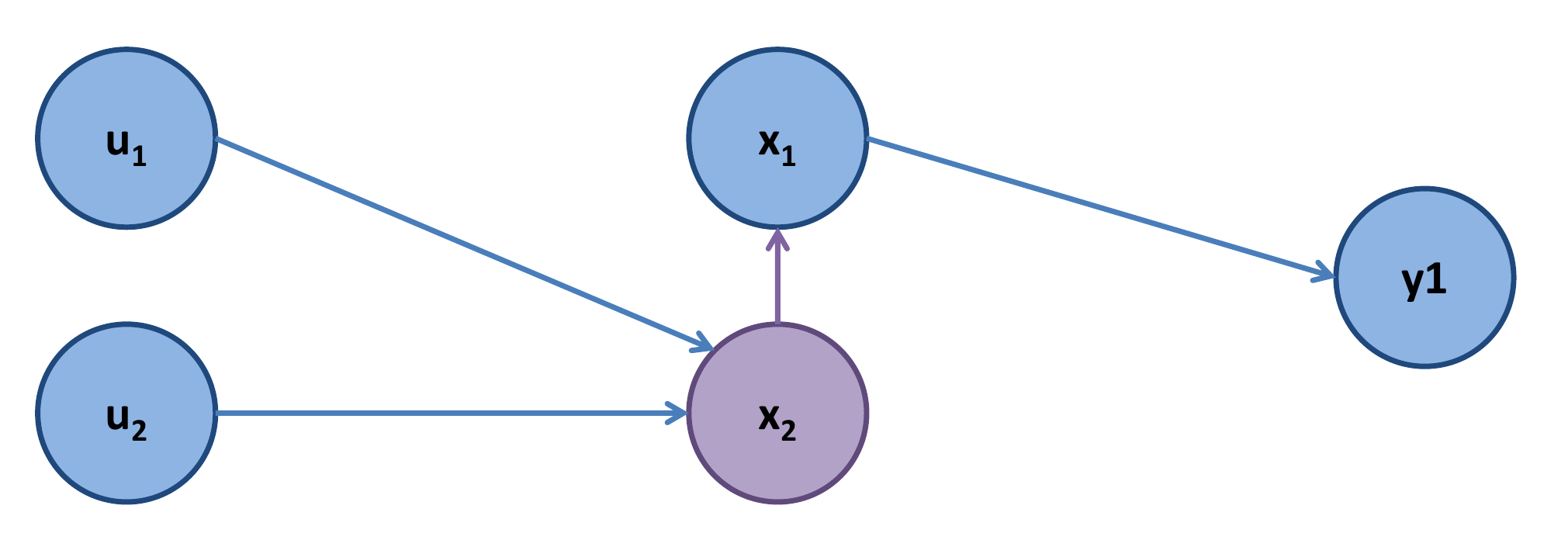}
                \caption{Complete computational structure with shared hidden node.}
                \label{fig:ccswshared}
        \end{subfigure}
        \begin{subfigure}[b]{0.4\textwidth}
                \includegraphics[page=3,width=\textwidth]{sharedhiddenstate}
                \caption{Signal Structure containing paths that represent two separate systems.}
                \label{fig:sigstrshared}
        \end{subfigure}
        \begin{subfigure}[b]{0.4\textwidth}
                \includegraphics[page=2,width=\textwidth]{sharedhiddenstate}
                \caption{Subsystem Structure containing only a single subsystem.}
                \label{fig:substrshared}
        \end{subfigure}
        \caption{Shared Hidden State}\label{fig:sharedhiddenstate}
\end{figure}





The dynamical structure function of a system is uniquely defined given a zero-intricacy state space model, as derived previously in this section. Determining a unique state space model given a system's dynamical structure function is an ill-posed problem, though a procedure for determining a minimal state space realization given a dynamical structure function $(Q,P)$ was given in \cite{minqp}.

\section{Applications of Dynamical Structure Functions}

The dynamical structure function is a versatile system representation and in this section we develop several features of the dynamical structure function including: network reconstruction, vulnerability analysis, and distributed controller design.

\subsection{Network Reconstruction}

\label{sec:netrecon}

Network reconstruction is the process of determining the structure of an unknown system, \cite{nessandsuff}. The network reconstruction process detailed here determines a unique dynamical structure function given a system's transfer function $H$. Given $H$ and noting the relationship in (\ref{eq:tfdsf1}) we can reorder the equation to get
\begin{equation}
\label{eq:start}
\left[\begin{array}{cc}I & H^{T}\end{array}\right] \left[\begin{array}{c}P^{T}\\Q^{T}\end{array}\right] = H^{T}
\end{equation}
\noindent where $A^T$ represents the transpose of $A$.
Noting that 
\[
 AX=B \iff {\rm blckdiag(A,...,A)}\overleftarrow{x}=\overleftarrow{b}
 \]
\noindent where $blckdiag(A,..., A) = \begin{bmatrix} A & 0 & 0 \\ 0 & \ddots & 0 \\ 0 & 0 & A \end{bmatrix}$ and $\overleftarrow{a}$ is a vector of the stacked columns of the matrix $A$ and defining $X=\left[\begin{array}{cc}P^{T}&Q^{T}\end{array}\right]$, Equation (\ref{eq:start}) can be rewritten as
\begin{equation}
\label{eq:next}
\left[\begin{array}{cc}I & {\rm blckdiag}(H^T,...,H^T)\end{array}\right]\overleftarrow{x} =\overleftarrow{h}.
\end{equation}

Since the diagonal elements of $Q$ are identically zero and the dimensions of $P$, $Q$, and $H$ are $p\times m$, $p\times p$, and $p\times m$ respectively, where $p$ is the number of outputs from the system and $m$ is the number of inputs, then exactly $p$ elements of $\overleftarrow{x}$ are always zero.  Abusing notation, $\overleftarrow{x}$ can be defined with these zero elements removed, reducing Equation (\ref{eq:next}) to
\begin{equation}
\label{eq:final}\left[\begin{array}{cc}I & {\rm blckdiag}(H_{-1}^T, H_{-2}^T,..., H_{-p}^T)\end{array}\right]\overleftarrow{x} =\overleftarrow{h}.
\end{equation}
\noindent where $H_{-i}^T$ is the matrix $H^T$ with the $i^{th}$ column removed.

Equation (\ref{eq:final}) reveals the mapping from elements of the dynamical structure function, contained in $\overleftarrow{x}$, to its associated transfer function, represented by $\overleftarrow{h}$. The mapping is a linear transformation represented by the matrix operator $L =\left[\begin{array}{cc}I & {\rm blckdiag}(H_{-1}^T, H_{-2}^T,..., H_{-p}^T)\end{array}\right]$.  This matrix has dimensions $(pm) \times (pm+p^2-p)$, and, thus, is not injective. This means the problem of network reconstruction from input-output dynamics is ill-posed and other information about the system is required a priori in order to determine a unique dynamical structure function.

Certain elements of the vector $\overleftarrow{x}$ need to be known a priori in order to reduce the transformation to an injective map. To accomplish this, consider the $(pm+p^2-p)\times k$ transformation $T$ such that 
\begin{equation}
\label{eq:T}
\overleftarrow{x}=Tz
\end{equation}
where $z$ is an arbitrary vector of size $k$.  

Letting $M=LT$, which makes $M$ a $pm \times k$ matrix, $M$ will be injective if and only if ${\rm rank}(M) = k$, i.e. $M$ has full column rank. Observing that $M$ is the mapping from unidentified model parameters to the system's transfer function we see that if $M$ is injective, one can clearly solve for $z$ given $H$ and then construct the dynamical structure function from $\overleftarrow{x}$. \textbf{This means that a $T$ that ensures the rank of $M$ is equal to $k$ is precisely the a priori system information that is necessary and sufficient for reconstruction of a unqiue dynamical structure function given a system's transfer function.}

\begin{example}
Given the following transfer function of a system
\[H = \begin{bmatrix} \frac{s+3}{s^2 + 3s +2} & -\frac{(s+3)}{s^3+6s^2+10s+5} \\ \frac{1}{s+1} & -\frac{(s^2+5s+6)}{s^3+6s^2+10s+5} \end{bmatrix}\]
we attempt to find the dynamical structure function $(Q,P)$ of the system
\[Q = \begin{bmatrix} 0 & Q_{12} \\ Q_{21} & 0 \end{bmatrix} \text{ and } P = \begin{bmatrix} P_{11} & P_{12} \\ P_{21} & P_{22} \end{bmatrix}\]
yielding the vector of unknowns $\vec{x} = \begin{bmatrix} P_{11} & P_{12} & P_{21} & P_{22} & Q_{12} & Q_{21}\end{bmatrix}'$. 
This gives us $L\vec{x} = \vec{b}$ as
\[ \begin{bmatrix} 1 & 0 & 0 & 0 & \frac{1}{s+1} & 0\\ 
		     0 & 1 & 0 & 0 & -\frac{(s^2+5s+6)}{s^3+6s^2+10s+5} & 0\\
		     0 & 0 & 1 & 0 & 0 & \frac{s+3}{s^2 + 3s +2}\\
		     0 & 0 & 0 & 1 & 0 & -\frac{(s+3)}{s^3+6s^2+10s+5} \end{bmatrix}\begin{bmatrix} P_{11} \\ P_{12} \\ P_{21} \\ P_{22} \\ Q_{12} \\ Q_{21}\end{bmatrix} = 
\begin{bmatrix} \frac{s+3}{s^2 + 3s +2} \\ -\frac{(s+3)}{s^3+6s^2+10s+5} \\ \frac{1}{s+1} \\ -\frac{(s^2+5s+6)}{s^3+6s^2+10s+5} \end{bmatrix}
\]

Without additional information a priori structural information, the system can not be reconstructed. Suppose, however, that a priori information details that $P$ has the form 
\[P = \begin{bmatrix} P_{11} & 0 \\ P_{21} & -P_{11} \end{bmatrix}.\]
Using this information the vector $\vec{x}$ can then be decomposed into the form $T\vec{z}$ where
\[T = \begin{bmatrix}1 & 0 & 0 & 0 \\ 0 & 0 & 0 & 0 \\ 0 & 1 & 0 & 0 \\ -1 & 0 & 0 & 0 \\ 0 & 0 & 1 & 0 \\ 0 & 0 & 0 & 1 \end{bmatrix} \text{ and } \vec{z} = \begin{bmatrix} P_{11} \\ P_{21} \\ Q_{12} \\ Q_{21} \end{bmatrix}\]

Replacing $\vec{x}$ with $T\vec{z}$ above yields $M\vec{z} = \vec{b}$, where $M=LT$, as
\[
\begin{bmatrix} 1 & 0 &  \frac{1}{s+1} & 0\\ 
		     0 & 0 & -\frac{(s^2+5s+6)}{s^3+6s^2+10s+5} & 0\\
		     0 & 1 & 0 & \frac{s+3}{s^2 + 3s +2}\\
		     -1 & 0 & 0 &  -\frac{(s+3)}{s^3+6s^2+10s+5} \end{bmatrix}\begin{bmatrix} P_{11} \\ P_{22} \\ Q_{12} \\ Q_{21}\end{bmatrix} 
= \begin{bmatrix} \frac{s+3}{s^2 + 3s +2} \\ -\frac{(s+3)}{s^3+6s^2+10s+5} \\ \frac{1}{s+1} \\ -\frac{(s^2+5s+6)}{s^3+6s^2+10s+5} \end{bmatrix}
\] In this case $M$ is full rank, which means that the system is reconstructible. By solving for $\vec{x} = M^{-1}\vec{b}$ we get the dynamical structure function
\begin{equation}Q = \begin{bmatrix} 0 & \frac{1}{s+2} \\ \frac{1}{s+3} & 0\end{bmatrix} \text{ and } P = \begin{bmatrix} \frac{1}{s+1} & 0 \\ \frac{1}{s+2} & -\frac{1}{s+1}\end{bmatrix}\label{eq:dsfnetreconex}\end{equation}
The signal structure corresponding to the dynamical structure function in (\ref{eq:dsfnetreconex}) is given in Figure \ref{fig:netrecssex}.

\begin{figure}[h!] \centering \includegraphics[width=.7\textwidth]{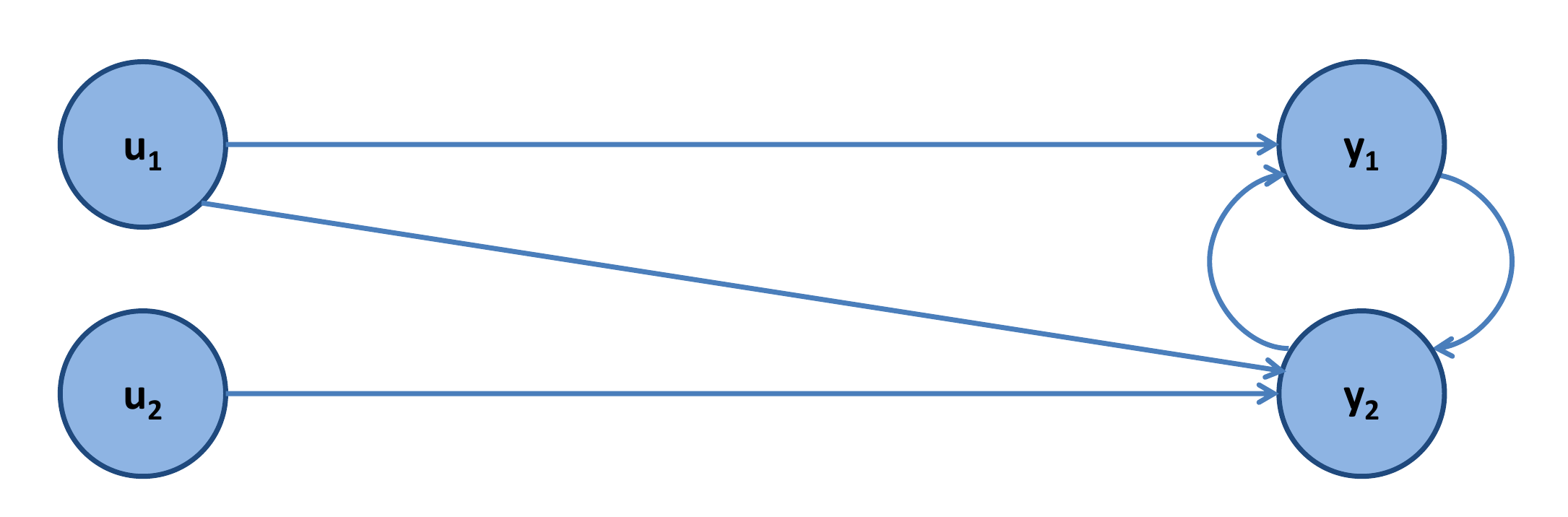} \caption{Reconstructed Signal Structure} \label{fig:netrecssex} \end{figure}

\end{example}

Note that robust reconstruction methods for reconstructing the dynamical structure function of a system in the face of noise and nonlinearities were first developed in \cite{robust}. A more efficient polynomial time algorithm for network reconstruction from noisy data was developed in \cite{robpoly}. Further improvements to the robust reconstruction process, including a more efficient algorithm that also allowed for the reconstruction of more systems, were detailed in \cite{polynomial}.

This network reconstruction process is useful because it allows us to determine a notion of structure from a representation with little structural information under very reasonable conditions. Once the structure of the system is determined, more analysis can be performed on the system, such as the vulnerability of the system to various classes of attacks.

\subsection{Vulnerability Analysis}
\label{sec:vulnerability}

Vulnerability analysis on a system attempts to determine the effects of external attacks or accidental component failures on the overall dynamics of the system. Conducting a vulnerability analysis of a system is important in the design stages of the system construction, because it allows for the system to be built to be robust to both internal and external perturbations. 

In order to conduct a vulnerability analysis on a system's dynamical structure function, the derivation of the dynamical structure function must be extended to include external disturbances. This can be done by first redefining the state space model (\ref{eq:ssequation}) to include an external disturbance term
\begin{equation} \begin{array}{rcl} \begin{bmatrix}\dot{y} \\ \dot{z}\end{bmatrix} & = & \begin{bmatrix} A_{11} & A_{12} \\ A_{21} & A_{22} \end{bmatrix}\begin{bmatrix}y \\ z\end{bmatrix} + \begin{bmatrix} B_1 \\ B_2 \end{bmatrix} u + \begin{bmatrix} F_1 \\ F_2 \end{bmatrix} \psi \\ y & = & \begin{bmatrix}I & 0\end{bmatrix} \begin{bmatrix}y \\ z\end{bmatrix} \end{array} \label{eq:ss} \end{equation}

Following a similar process to that found in \cite{dsfpaper} we find the corresponding dynamical structure function to be
\begin{equation} Y = QY + PU + \Delta\Psi \label{eq:newdsf} \end{equation} where \[Q = (sI-D)^{-1}(W-D)\] \[P = (sI-D)^{-1}V\] \[\Delta = (sI-D)^{-1}N\] with $W = A_{11} + A_{12}(sI-A_{22})^{-1}A_{21}$, $D = diag(W_{11}, ..., W_{pp})$, $V = B_{1} + A_{12}(sI-A_{22})^{-1}B_{2}$, and $N = F_{1} + A_{12}(sI-A_{22})^{-1}F_{2}$. Equation (\ref{eq:newdsf}) is then a generalized attack model in the dynamical structure function domain.

Focusing on a class of destabilizing attacks, assume that the system being analyzed is stable and only consider attack models that use the existing communication structure to conduct an attack. This is not an unreasonable assumption since creating new links within a system may be a difficult or expensive task for an attacker.

\subsubsection{Vulnerability of a Single Link Attack} \label{sec:vulnsinglelink}

Starting with (\ref{eq:newdsf}) and solving for $Y$ in terms of $U$ and $\Psi$ yields \begin{equation}Y = (I-Q)^{-1}PU+(I-Q)^{-1}\Delta\Psi \label{eq:tfdsf}\end{equation} \noindent where the input-output relationship is given by $H = (I-Q)^{-1}P$ and the transfer function describing how $\Psi$ affects the exposed states, $Y$, is $(I-Q)^{-1}\Delta$. Given that the system is stable, no bounded input can destabilize the system, so the case when $\Psi = Y$ is analyzed since it means that an attacker is using some combination of additive perturbations on existing communication links to destabilize the system. 

In \cite{anurag}, it states that a stable additive perturbation $\Delta$ on a link $Q_{ij}$ or $P_{ij}$ is able to destabilize the system if and only if the transfer function, $M_{ij}$, seen by $\Delta$ is nonzero. This means that the link $Q_{ij}$ or $P_{ij}$ is in feedback with some series of links in $Q$ or $P$. Note that $M_{ij}$ is the transfer function from $\Delta Y_j$ to $Y_{j}$. The transfer function seen by a perturbation $\Delta Y$ is then given in (\ref{eq:tfdsf}) as $(I-Q)^{-1}$. In particular, if we want to determine the vulnerability of a single link attack on a link $Q_{ij}$, we know this can be modeled as $\Delta_{ij} = (sI-D_{ii})^{-1}N_{ij}$ with the rest of the entries in $\Delta$ equal to zero. Then, the transfer function seen by the perturbation on the link $Q_{ij}$ is found from
\begin{equation}\begin{bmatrix}Y_1 \\ \vdots \\ Y_{j-1} \\ Y_{j} \\ Y_{j+1} \\ \vdots \\ Y_{p} \end{bmatrix} = K\begin{bmatrix} 0 \\ \vdots \\ 0 \\ \Delta_{ij}Y_j \\ 0 \\ \vdots \\ 0 \end{bmatrix}\label{eq:hij}\end{equation}
\noindent where $K = (I-Q)^{-1}$. From (\ref{eq:hij}), we can see that $Y_j = K_{ji}\Delta_{ij}Y_j$ since $\Delta_{ij}Y_j$ is in the $i^{th}$ row of the vector given in (\ref{eq:hij}). Therefore, the vulnerability of a single link can be defined as
\[v_{ij} = ||K_{ji}||_\infty\]
\noindent which means that the vulnerability of the entire system is
\begin{equation}V = \max_{Q_{ij} \neq 0 \in Q} ||K_{ji}||_\infty\end{equation}
which is simply the maximum possible vulnerability across all links.

\subsubsection{Vulnerability of a Multiple Link Distributed Attack}

Consider now an attack in which multiple attackers are simultaneously performing unique single link attacks in the system and are not sharing information, known as a distributed attack. This is modeled by the concatenation of several single link attacks on the system and by application of the small gain theorem, the vulnerability, $v_{ij,...,kl}$ of this type of an attack is the structured singular value, $\mu_{ij,...,kl}$, of the matrix
\begin{equation} R_{ij,...,kl} = 
\begin{bmatrix} K_{ji} & 0 & ... & 0 \\ 0 & \ddots & \ddots& \vdots \\ \vdots & \ddots &\ddots & 0 \\ 0 & ... & 0 & K_{lk} \end{bmatrix}\label{eq:R}\end{equation}
That is,
\[v_{ij,...,kl} = \mu(R_{ij,...,kl}, \Pi)\]

The overall vulnerability of the system to a distributed attack is
\[V = \max_{R_{links} \in \mathscr{R}} \mu_{links}\]
\noindent where $\mathscr{R}$ is the set of matrices of the form (\ref{eq:R}) over the set of all possible combinations of links, $\mathscr{L}$, and $\mu_{links}$ is the structured singular value of $R_{links}$.

\subsubsection{Vulnerability of a Multiple Link Co-ordinated Attack}

A multiple link co-ordinated attack is another generalizaation of a single link attack and is similar to a distributed attack, except that it models either communication between multiple attackers or a single attacker targeting multiple links. The transfer function seen by a perturbation on multiple links when allowing for communication in the attack is then given by
\begin{equation} 
T_{ij,...,kl} = \begin{bmatrix}K_{ij} & ... & K_{il} \\ \vdots & \ddots & \vdots \\ K_{kj} & ... & K_{kl} \end{bmatrix} \label{eq:M}
\end{equation}

In this case, the vulnerability of a multiple link co-ordinated attack is
\[v_{ij,...,kl} = ||T_{ij,...,kl}||_{\infty}\]
\noindent and the overall vulnerability of the system to a co-ordinated attack is
\[V = \max_{links \in \mathscr{L}} ||T_{links}||_\infty \]
 
\subsubsection{Reducing Vulnerability in Open-Loop Systems}
\label{sec:minimizingvulnopenloop}

Since the vulnerability of any given link in a system is the transfer function seen by a perturbation on that link, the vulnerability of a system to bounded perturbations is nonzero if and only if feedback exists within the system. Therefore, a completely secure architecture is one in which no links in $Q$ exist. 

Note that since $G=(I-Q)^{-1}P$, when $Q = 0$, then $P = G$. Since links in $P$ are never in feedback for open-loop systems in which attackers cannot create links, they are never vulnerable (see Example \ref{ex:vuln}). Thus, the overall vulnerability of a system with $Q = 0$ is $V = 0$, meaning there does not exist a finite additive perturbation on a link in the system that can destabilize the system under the assumption that the attacker can only use the existing communication network of the system \cite{anurag}.

\begin{example}
\label{ex:vuln}
Consider the following dynamical structure function
\begin{equation}Q = \begin{bmatrix}0 & \frac{1}{s+1} & 0 \\ 0 & 0 & \frac{1}{s+2} \\ \frac{1}{s+3} & 0 & 0\end{bmatrix}, \text{   } P = \begin{bmatrix} \frac{1}{s+1} & 0 & 0 \\ 0 & \frac{1}{s+1} & 0 \\ 0 & 0 & \frac{1}{s+1}\end{bmatrix}\label{eq:vulnqp}\end{equation}

The corresponding transfer function for (\ref{eq:vulnqp}) is 
\begin{equation} G = \frac{1}{s^3+6s^2+11s+5}\begin{bmatrix} s^2+5s+6 & \frac{s^2+5s+6}{s+1} & \frac{s+3}{s+1} \\ 1 & s^2+5s+6 & s+3 \\ s+2 & \frac{s+2}{s+1} & s^2+5s+6 \end{bmatrix} \end{equation} By the small gain theorem, the smallest perturbation that could destabilize the system is $\frac{1}{||G||_\infty} = 0.4152$, which makes the vulnerability of the input-output system $V = ||G||_{\infty} = 2.4085$.

The signal structure of the system in (\ref{eq:vulnqp}), pictured in Figure \ref{fig:vulnarc}, is a ring structure with a feedback loop in $Q$. As mentioned in Section \ref{sec:minimizingvulnopenloop}, this system structure is vulnerable to destabilizing attacks that target specific links. Given $H = (I-Q)^{-1}$, we can determine the vulnerabilities of each link $Q$:
\[\begin{array}{rcl} v_{12} & = & .2 \\ v_{23} & = & .4 \\ v_{31} & = & .6 \end{array}\]
The overall vulnerability of the system to single link attacks is $V = v_{31} = .6 < ||G||_{\infty}$, which means the smallest perturbation on a single link that can destabilize the system is $\frac{1}{V} = 1.67$. This perturbation is smaller than the perturbation required to destabilize the input-output representation since it restricts attacks to only perturb one link within the system, rather than perturbations that affect the entire system.

\begin{figure}[h!] \centering \includegraphics[page=1,width=.7\textwidth]{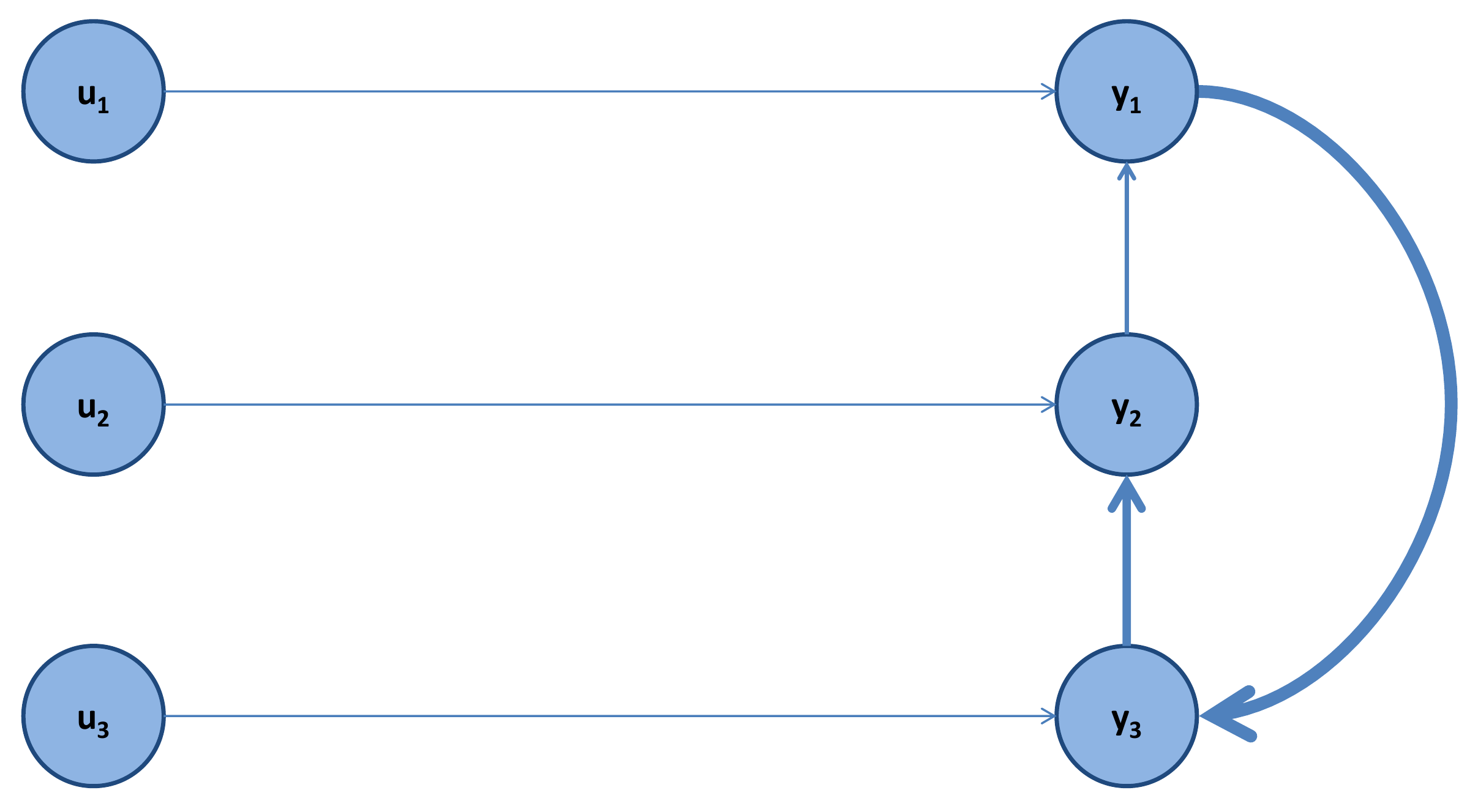} \caption{Vulnerable structure: the size of each link corresponds to its vulnerability. Remember that links in $P$ are not vulnerable.} \label{fig:vulnarc} \end{figure}

Also discussed Section \ref{sec:minimizingvulnopenloop} is that one possible secure structure for the system in (\ref{eq:vulnqp}) is one in which we remove all the links in $Q$, which removes all feedback from $Q$. This system would have a dynamical structure function of the form
\begin{equation} \bar{Q} = \begin{bmatrix} 0 & 0 & 0 \\ 0 & 0 & 0 \\ 0 & 0 & 0 \end{bmatrix}, \text{   } \bar{P} = \frac{1}{s^3+6s^2+11s+5}\begin{bmatrix} s^2+5s+6 & \frac{s^2+5s+6}{s+1} & \frac{s+3}{s+1} \\ 1 & s^2+5s+6 & s+3 \\ s+2 & \frac{s+2}{s+1} & s^2+5s+6 \end{bmatrix}\label{eq:secqp} \end{equation}

The signal structure of the system in (\ref{eq:secqp}), pictured in Figure \ref{fig:secarc}, has no feedback in $Q$ while still maintaining all the pathways from inputs to measured states/outputs that existed in the original system from (\ref{eq:vulnqp}), making it a secure structure without compromising the input-output dynamics of the system. Note that by \emph{secure} we mean that the vulnerability of the system to both single link and multiple link attacks is $0$, although the system may still be vulnerable to other types of attacks. 

\begin{figure}[h!] \centering \includegraphics[page=2,width=.7\textwidth]{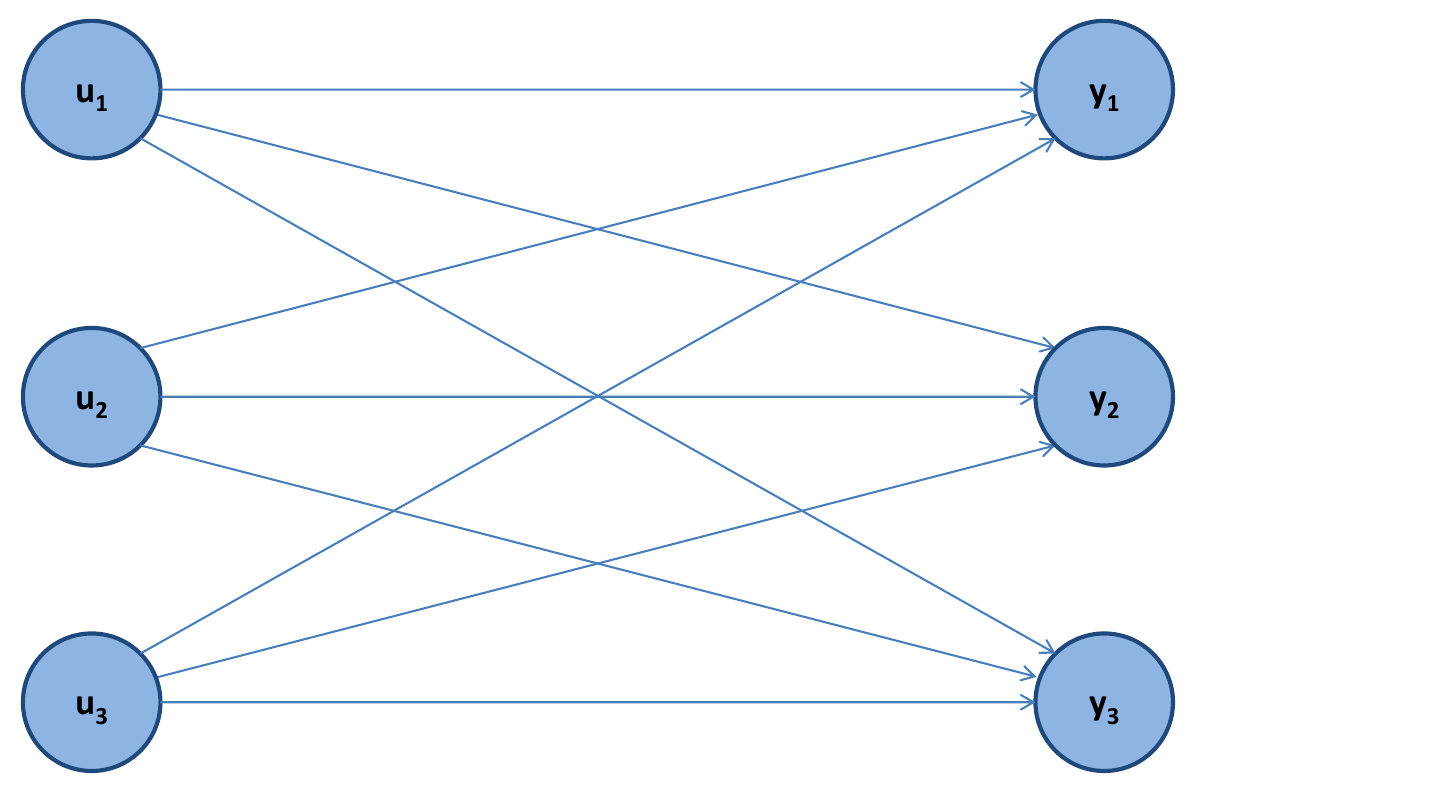} \caption{Secure structure: $P=G$ and $Q=0$, meaning that no links in the system are in feedback, so no links in the system are vulnerable to single link attacks that use the existing communication network.} \label{fig:secarc} \end{figure}

Unfortunately, in some systems removing all internal links is infeasible and, in some cases, feedback in a system is necessary. Consider an unstable plant in feedback with a stabilizing controller, the feedback is necessary to keep the system stable, but this produces vulnerable links within the closed loop system. The process of minimizing vulnerability in the face of feedback is still an open problem, although preliminary work in this area was conducted in \cite{nathansthesis}.

\end{example}

\subsection{Structured Controller Design}
\label{sec:distcontdes}

Another interesting problem involving stabilizing an unstable system is the problem of structured controller design, which refers to the design of a stabiizing controller where the structure of the controller is restricted by certain conditions. For example, a controller could be restricted to using existing communication links within a system or the controller must be designed to reduce the vulnerability of links to an attack. Denoting $(Q^{bin}, P^{bin})$ to be the boolean structure of allowable links in a controller and $H$ the transfer function of the unstable plant, the following procedure was developed in \cite{stabcont} for generating a stabilizing controller with the desired structure \\ \\
\textbf{Procedure }$\mathbb{P}$
\begin{enumerate}
\item Choose an undesigned link $p_{ij}$ such that $p_{ij}^{bin}=1$.
\item Design $p_{ij}$ to stabilize $h_{ji}$ such that there is no pole zero cancellation in $PG$. That is, the controller link is designed such that it stabilizes the transfer function it sees, and there is no pole-zero cancellation.
\item After adding $p_{ij}$, if the closed loop system $(H,P)$ is still unstable, repeat for all $p_{xy},$ where $p_{xy}^{bin}=1$, so that the added link attempts to stabilize the plant as well as all the previously added controller links.
\item If the closed loop system $S$, formed by adding $P$ in feedback with $H$, is still unstable, add links in $Q^{bin}$ such that there is no pole-zero cancellation between $Q$ and $S$. Again, each added link attempts to stabilize the plant $H$ along with the previously added links of $P$ and $Q$.
\end{enumerate}

Furthermore, it was shown in \cite{stabcont} that if this procedure does not create a stabilizing controller, then no such controller given the restrictions $(Q^{bin}, P^{bin})$ exists in the system. In particular, we note that this procedure works if the unstable modes of the plant in the system is both observable and controllable by the controller with the required structure. So, for example, a controller with a diagonal structure, i.e. $P^{bin} = \begin{bmatrix} 1 & 0 & ... & 0 \\ 0 & 1 & \ddots & \vdots \\ \vdots & \ddots & \ddots & 0 \\ 0 & ... & 0 & 1 \end{bmatrix}$ and $Q^{bin} = 0$, which represents a completely distributed controller, can only stabilize a system if the system's unstable modes that are controllable from input $i$ are also observable from output $i$. 

Another interesting controller structure is the cycle structure
\[P^{bin} = \begin{bmatrix} 1 & 0 & ... & 0 \\ 0 & 1 & \ddots & \vdots \\ \vdots & \ddots & \ddots & 0 \\ 0 & ... & 0 & 1 \end{bmatrix} \text{ and } Q^{bin} = \begin{bmatrix} 0 & 1 & 0 & ... & 0 \\ \vdots & \ddots &\ddots&\ddots & \vdots \\ \vdots & & \ddots & \ddots & 0 \\ 0 & ... & ... & 0 & 1 \\ 1 & 0 & ... & ... & 0 \end{bmatrix}\] which can stabilize any unstable plant as long as the plant is detectable and stabilizable, since the cyclic structure of the controller allows every unstable mode to be observable from any output and controllable from any input.

\begin{example}
\label{ex:distcont}
Given the following unstable system
\[\begin{array}{rcl}\dot{x} & = & \begin{bmatrix}1 & 0 & 0 \\ 1 & 2 & 3 \\ 1 & 0 & 3 \end{bmatrix}x + \begin{bmatrix} 1 & 0 \\ 0 & 1 \\ 0 & 0 \end{bmatrix}u \\ y & = & \begin{bmatrix} 1 & 0 & 0 \\ 0 & 1 & 0 \end{bmatrix}x \end{array}\]
with associated dynamical structure function
\[Q = \begin{bmatrix} 0 &  0 \\ \frac{s}{s^2-5s+6} & 0 \end{bmatrix}, \text{   } P = \begin{bmatrix}\frac{1}{s-1} & 0 \\ 0 & \frac{1}{s-2}\end{bmatrix}\]
with a signal structure as shown in Figure \ref{fig:unpl},
\begin{figure}[h!] \centering \includegraphics[page=1,width=.65\textwidth]{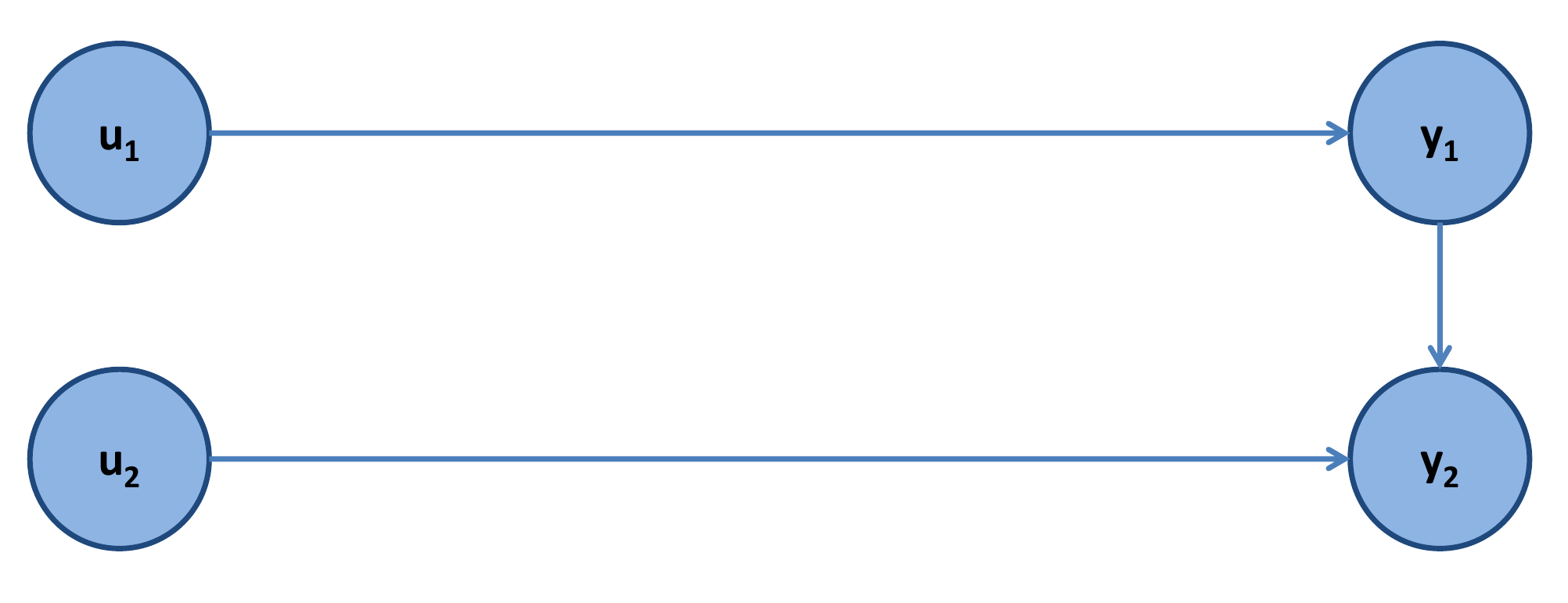} \caption{Unstable plant structure} \label{fig:unpl} \end{figure}
restrict the controller structure to be a diagonal controller, shown in Figure \ref{fig:diagpl}.
\begin{figure}[h!] \centering \includegraphics[page=2,width=.65\textwidth]{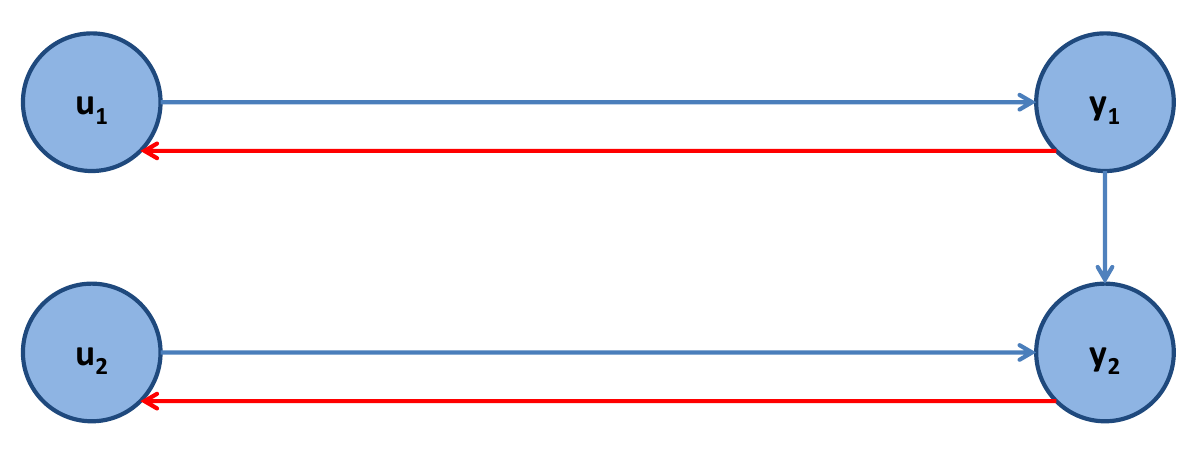} \caption{Unstable plant with diagonal controller structure in red} \label{fig:diagpl} \end{figure}

Noting that the system has modes $\{1,2,3\}$, the Popov-Belevitch-Hautus test for controllability and observability shows that mode $3$ is controllable only from input $u_1$, but observable only from output $y_2$, which means no stabilizing diagonal controller exists. However, since the system is both stabilizable and detectable, the system can be stabilized by a cyclic controller, shown in Figure \ref{fig:cyclpl}.

\begin{figure}[h!] \centering \includegraphics[page=3,width=.65\textwidth]{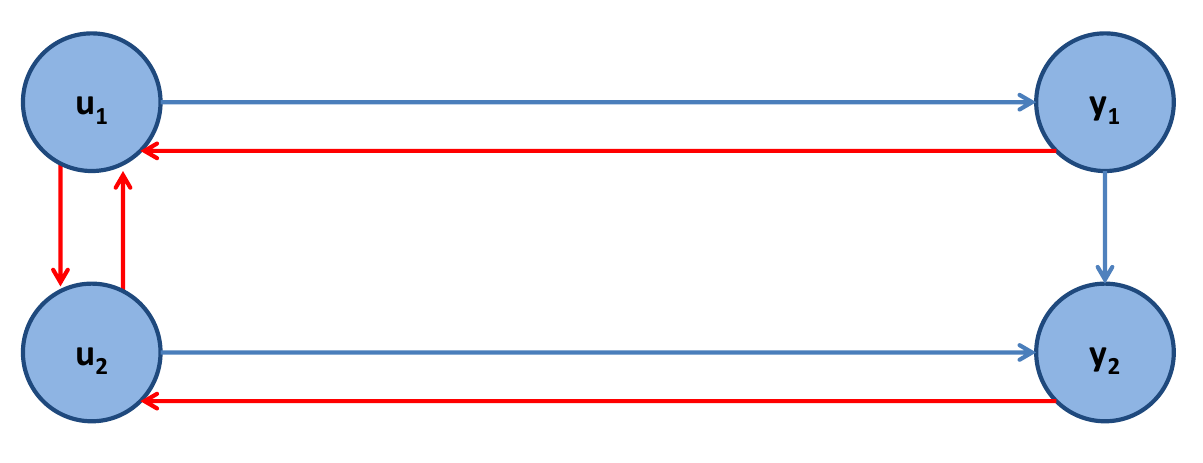} \caption{Unstable plant with stabilizing cyclic controller structure in red} \label{fig:cyclpl} \end{figure}

\end{example}

\section{Conclusion}

This chapter explored four different mathematical system representations and their associated structures.  Three of these representation are standard for LTI systems: transfer functions, state space models, and the interconnection of subsystems.  The fourth representation, the dynamical structure function, and its associated structure, the signal structure, are relatively new.

The dynamical structure function and its signal structure were then used to discuss three important problems: network reconstruction, vulnerability analysis, and the design of distributed stabilizing controllers.  These applications highlight the practicality of a theory of structures for networks of dynamic systems.  

\bibliographystyle{unsrt}
\bibliography{paper_refs}
\end{document}